\documentclass[11pt]{article}
\usepackage[autonum,bbold,colorhypersetup]{shortex}
\usepackage{url}
\usepackage{amsfonts}
\usepackage{amsmath}
\usepackage{amssymb}
\usepackage{amsthm}
\usepackage{array}
\usepackage{bm}
\usepackage{booktabs}
\usepackage[dvipsnames]{xcolor}
\usepackage{easybmat}
\usepackage{graphicx}
\usepackage{times}
\usepackage[round]{natbib}
\usepackage[margin=1in]{geometry}
\usepackage{setspace}
\makeatletter
\def\blfootnote{\gdef\@thefnmark{}\@footnotetext}
\makeatother

\title{Margin-closed vector autoregressive time series models}
\author{\normalsize
Lin Zhang$^{1\ast}$\qquad
Harry Joe$^{1}$\qquad
Natalia Nolde$^{1}$\\
{\footnotesize  $^1$\textit{Department of Statistics, University of British Columbia,
Vancouver, BC Canada V6T 1Z4}}
}
\date{}

\begin{document}
\blfootnote{$^{\ast}$Corresponding author. Email address: lin.zhang@stat.ubc.ca}
\maketitle

\begin{abstract}

Conditions are obtained for a Gaussian vector autoregressive time series of order $k$, VAR($k$), to have
univariate margins that are autoregressive of order $k$ or lower-dimensional margins that
are also VAR($k$). This can lead to $d$-dimensional VAR($k$) models that are closed with respect to a given partition $\{S_1,\ldots,S_n\}$ of
$\{1,\ldots,d\}$ by specifying marginal serial dependence and some cross-sectional dependence parameters.
The special closure property allows one to fit the sub-processes of multivariate time series before assembling them by fitting the dependence structure between the sub-processes. We revisit the use of the Gaussian copula of the stationary joint distribution of observations in the VAR($k$) process with non-Gaussian univariate margins but under the constraint of closure under margins. This construction allows more flexibility in handling higher-dimensional time series and a multi-stage estimation procedure can be used.
The proposed class of models is applied to a macro-economic data set and compared with the relevant benchmark models.

\medskip

\end{abstract}

\section{Introduction}
\label{sec:1}

The stationary Gaussian vector autoregressive (VAR) model of order $k$ is a common model
for multivariate time series; see \cite{LutkepohlS2005}  and
\cite{Wei2006}.
It is a Markov model of order $k$, but it does
not imply that its univariate marginal time series are Markov of order $k$
or autoregressive with finite order. In the copula and other multivariate
literature, it is common to start with univariate models and then add
dependence to model the multivariate relations.
This is important because there are many diagnostics that help in selecting
univariate models.

Despite extensive literature on VAR models, the conditions for when
$d$-variate stationary Gaussian  VAR($k$) time series have univariate AR($k$) margins
or lower-dimensional  VAR($k$) margins have not been previously studied.
It is clear that the marginal closure property is satisfied if the
$k$ coefficient matrices are diagonal, and the dependence comes only from
the innovations for the $d$ univariate time series. More generally,
one partial univariate margin of the VAR($k$) time series is AR($k$)
if the corresponding rows of the coefficient matrices are all zero except
the diagonal element.
The main goal of this paper is to obtain conditions for margins to
be AR($k$) or VAR($k$), even if the coefficient matrices are not diagonal
and do not have any row with only one non-zero element.
The conditions are obtained via some conditional independence relations
that have to be satisfied for the marginal closure property.

Our motivation comes from multivariate time series with variables that
do not have Gaussian distributions. Assuming stationarity, univariate models
can be chosen and then fitted, followed by transformation to standard
normal. If the times series on the transformed scale have diagnostics
that suggest AR($k$), with $k$ as the largest order over $d$ univariate time series,
then a simple idea of building dependence between $d$ univariate time series
is to try Gaussian VAR($k$) dependence as a copula for the $d$ univariate time series.
In this case, we can make use of the conditions for closure under margins to have a copula model with
fewer parameters than VAR($k$) without the closure property.
In order to use other non-Gaussian copulas for the $d$ univariate time series,
more research is needed to meet the conditions for both closure under margins and strict stationarity.

The use of copulas from Gaussian VAR to represent dependence has been considered in
\cite{biller2003modeling} and  \cite{biller2009copula}, which leads to multivariate time series with non-Gaussian margins.
However, these authors do not consider parsimonious forms resulting from closure under margins restrictions.

The special property of closure under margins enables the new framework of modeling high-dimensional time series by modeling its low-dimensional sub-processes and then assembling them.
The framework makes it more flexible in analyzing marginal behavior of
the multivariate time series, and including additional marginal components.
It also enables a multi-stage estimation procedure.

The rest of this article is organized as follows.
Section \ref{sec:VAR} introduces notation for different parameterizations of VAR models
and gives an overview of stationary multivariate time series models with multivariate Gaussian copulas.
Section \ref{sec:closure} discusses the property of closure under margins,
gives a sufficient condition for it to hold for VAR models,
and provides details of parameterizing the margin-closed multivariate time series models.
Section \ref{sec:estimation}  discusses estimation of the VAR models
assuming closure under margins.
Section \ref{sec:examples} presents some numerical examples and the results of an application of the model.
Section \ref{sec:conclusion} concludes the paper.
The Appendix include the proofs of the theorems in Section \ref{sec:closure},
as well as some details of solving linear systems to get the
property of closure under margins.

\section{Gaussian VAR Process}
\label{sec:VAR}

The stationary Gaussian VAR($k$) process is reviewed,
and it is shown how its stationary joint distribution can be used as a copula.

We start with a strictly stationary $d$-variate time series $\cbra{\bZ_t}_{t>0}$ where $\bZ_t = \rbra{Z_{1,t},\dots,Z_{d,t}}^T$.
Assuming that the variables have been centered to have zero means,
the VAR($k$) time series model can be expressed as:
\begin{equation} \label{eq:VAR(k)_model}
  \bZ_t = \bPhi_{1}\bZ_{t-1} + \cdots + \bPhi_{k}\bZ_{t-k} +
\bepsilon_{t},\quad \bepsilon_{t} \overset{\mathrm{i.i.d.}}{\sim}
 \mcN_d \rbra{\bm{0}, \bSigma_{\bepsilon}},
\end{equation}
where $\bPhi_{1},\dots,\bPhi_{k}$ are coefficient matrices and
$\bepsilon_{t}$ is the innovation vector whose covariance matrix is
$\bSigma_{\bepsilon}$. See \citep{LutkepohlS2005} for details and
properties.
If the coefficient matrices and innovation covariance matrix
are such that the process in \cref{eq:VAR(k)_model} is strictly stationary and Markov of order $k$,
the model can be characterized by its stationary joint distribution of $(k+1)$
consecutive observations.

To consider the use of the VAR($k$) model via the copula of
$\bZ_{t:(t-k)}=(\bZ^T_{t},\ldots,\bZ^T_{t-k})^T$ for $t>k$,
we suppose that $\bZ_{t}$ is a vector of dependent $\mcN(0,1)$
random variables with
\begin{equation}\label{eq: uni_trans_ZX}
  \bZ_t=\rbra{\Phi^{-1}(F_1(X_{1,t})),\ldots,\Phi^{-1}(F_d(X_{d,t}))}^T,
\end{equation}
where for $i=1,\ldots,d$, $\{X_{i,t}\}_{t>0}$ is a strictly stationary univariate time series
with $X_{i,t}\sim F_i$.

Then the cumulative distribution function (CDF) of the stationary joint distribution of $\rbra{\bX_t^T, \dots, \bX_{t-k}^T}^T$, which is denoted by $F_{\bX_{t:(t-k)}}$, can be represented as
\begin{equation}\label{eq:stationary_joint_cdf_by_copula}
  F_{\bX_{t:(t-k)}}\rbra{\bx_t, \dots, \bx_{t-k}; \boldeta_1, \dots,\boldeta_d, \bdelta}
  = C_{\bZ_{t:(t-k)}}\rbra{u_{1,t}, \dots, u_{d,t}, \dots, u_{1,t-k}, \dots, u_{d,t-k};\bdelta}
\end{equation}
where $C_{\bZ_{t:(t-k)}}$ is the copula of the stationary joint distribution of $\rbra{\bZ_t^T, \dots, \bZ_{t-k}^T}^T$, $\bdelta$ is the set of the dependence parameters, $\bx_{t-l}=\rbra{x_{1,t-l}, \dots, x_{d,t-l}}^T$ with $0\leq l \leq k$ is a realization of $\bX_{t-l}$, and $u_{i,t-l} = F_i\rbra{x_{i,t-l};\boldeta_i}$ for $1\leq i \leq d$.

Using standard results for copulas, such as in \cite{Joe2014},
the copula of $\bX_{t:(t-k)}$ is:
\begin{equation}\label{eq:multivariate_Gaussian_copula}
  C_{\bX_{t:(t-k)}}(\bv; \bdelta) = \Phi_{(k+1)d}\rbra{\Phi^{-1}(v_1), \dots, \Phi^{-1}(v_{(k+1)d}); R},
\end{equation}
where $\bv=(v_1,\dots,v_{(k+1)d})^T\in[0,1]^{(k+1)d}$ and
$\Phi_{(k+1)d}(\cdot,\dots,\cdot;R)$ is the joint CDF of the multivariate
Gaussian distribution with mean zero and correlation
matrix $R$, and $\bdelta$ is the vector of the VAR model parameters.
Note that $C_{\bX_{t:(t-k)}}(\cdot,\cdots,\cdot;\bdelta)$ is
the same function as $C_{\bZ_{t:(t-k)}}(\cdot,\cdots,\cdot;\bdelta)$
because copula is invariant with respect to strictly monotone transformations.
The block Toeplitz correlation matrix $R$ of dimension $(k+1)d\times (k+1)d$ can be
parameterized by VAR model parameter vector $\bdelta$ of dimension $d(d+1)/2+kd$.
Correlation matrices of multivariate Gaussian distributions
are restricted to be positive-definite.

A disadvantage of the stationary joint distribution specified in \cref{eq:stationary_joint_cdf_by_copula} is the dimension of the parameter set.
For multivariate Gaussian copulas, a parsimonious correlation matrix can be adopted
especially in the case of high-dimensional time series, high Markov order,
or moderate sample size. In subsequent sections,
we discuss a new approach to building multivariate time series models with parsimonious dependence structures by exploring the marginal models of Gaussian VAR processes.

\section{Closure under margins for Gaussian VAR model}
\label{sec:closure}

Section \ref{sec:example} explains the benefits of closure under
margins, and has a simple bivariate stationary VAR(1) example to show that
margins may not be univariate AR(1).
Section \ref{sec:condindep} specifies conditional independence relations
in order for a stationary VAR($k$) process to have marginal processes that
are also autoregressive of order $k$ or less.
Section \ref{sec:closure2} applies the results of Section \ref{sec:condindep}
to get conditions on the correlation matrix of $k+1$ consecutive
vectors and on the coefficient matrices of \cref{eq:VAR(k)_model}
in order to have closure under margins. The ideas are illustrated
through several examples.

\subsection{Property of closure under margins}
\label{sec:example}

For a multivariate model, being closed under margins refers to a property that any sub-vector or sub-process follows the same type of model as the original multivariate random vector or process.
A typical example of the property is the class of multivariate Gaussian
distributions, where any sub-vector has univariate or lower-dimensional multivariate
Gaussian distribution.
For VAR models, there are several advantages of the subclass with the closure under margins property.
First, the marginal models of any sub-processes of a multivariate time series can easily be obtained by
extracting the corresponding parameters of the VAR model that the original multivariate time series follows.
For fitting VAR models, if all univariate components of the multivariate time series follow AR models
with the same Markov order, the AR models of all univariate components can be fitted first, followed by estimating the cross-correlation parameters that are contemporaneous or lagged.
Even though \cite{Lutkepohl1984} indicated that the univariate components of
VAR models should follow ARMA models, this mathematical result is not of much practical use, because the bounds on the AR and MA orders can be wide.

To illustrate, let $\cbra{\rbra{Y_{1,t}, Y_{2,t}}}_{t>0}$ follows a bivariate
stationary VAR(1) process:
\begin{equation}\label{eq:example_of_VAR(1)}
  \begin{pmatrix}
    Y_{1,t} \\
    Y_{2,t}
  \end{pmatrix} =
  \begin{pmatrix}
    a_{11} & a_{12} \\
    a_{21} & a_{22}
  \end{pmatrix}
  \begin{pmatrix}
    Y_{1,t-1} \\
    Y_{2,t-1}
  \end{pmatrix} +
  \begin{pmatrix}
    \epsilon_{1,t} \\
    \epsilon_{2,t}
  \end{pmatrix}
\end{equation}
where $\epsilon_{1,t}$ and $\epsilon_{2,t}$ are independent $\mcN(0,1)$ random variables.
This process is closed under margins if and only if the two univariate component series, $\cbra{Y_{1,t}}_{t>0}$ and $\cbra{Y_{2,t}}_{t>0}$, follow AR(1) models.
In the case that $a_{11}=a_{21}=0$,
$-1<a_{12}<1$, and $-1<a_{22}<1$,
then Var$(Y_{2t})=(1-a_{22}^2)^{-1}$ and
the random vector
$\rbra{Y_{1,t}, Y_{1,t-1}, Y_{1,t-2}}^T$ follows a trivariate Gaussian distribution with zero mean and covariance matrix
\begin{equation}\label{eq:example_of_tri_cov}
  \frac{1}{1-a_{22}^2}
  \begin{pmatrix}
    1+a_{12}^2-a_{22}^2 & a_{12}^2a_{22}      & a_{12}^2a_{22}^2 \\
    a_{12}^2a_{22}      & 1+a_{12}^2-a_{22}^2 & a_{12}^2a_{22} \\
    a_{12}^2a_{22}^2    & a_{12}^2a_{22}      & 1+a_{12}^2-a_{22}^2
  \end{pmatrix},
\end{equation}
and the partial autocorrelation between $Y_{1,t}$ and $Y_{1,t-2}$ given $Y_{1,t-1}$ is
  $${a_{12}^2a_{22}^2[1-a_{22}^2] \over (1+a_{12}^2-a_{22}^2)^2-a_{12}^4a_{22}^2}. $$
Therefore, unless
$a_{12}=0$ or $a_{22}=0$,
the subprocess $\cbra{Y_{1,t}}_{t>0}$ does
not follow an AR(1) model, and \cref{eq:example_of_VAR(1)} is not closed under margins.
Indeed, as suggested in Section 11.6 in \cite{LutkepohlS2005}, $\cbra{Y_{1,t}}_{t>0}$ might not admit a finite Markov order VAR representation.
This shows that in general VAR models are not closed under margins.
However, the existence of the margin-closed VAR models can be verified in the special case when $a_{12}=a_{21}=0$,
$-1<a_{11}<1$ and $-1<a_{22}<1$
In this case, both subprocesses $\cbra{Y_{1,t}}_{t>0}$ and $\cbra{Y_{2,t}}_{t>0}$ follow AR(1) models, and the coefficient matrix is diagonal.

\subsection{Conditional independence relations for closure under margins}
\label{sec:condindep}

A VAR($k$) process is closed under margins if all its coefficient matrices are diagonal.
But, a margin-closed model with non-diagonal
coefficient matrices may give simpler interpretations in practice.
For example, non-zero values in the off-diagonal entries of coefficient matrices are meaningful when  Granger-Causality \citep{granger1969investigating} is investigated.
In this section, a sufficient condition under which a VAR model is closed under margins is derived.
Suppose the $d$ variables have been split into a partition of 2 to $d$ components, each with cardinality of at least 1.
Then, a VAR($k$) process is closed under margins with respect to this partition if and only if every sub-process in the partition is also a VAR($k$) process.
For example, a $3$-variate VAR($k$) process $\cbra{\rbra{Y_{1,t},Y_{2,t},Y_{3,t}}}_{t>0}$ is closed under margins with respect to partition $\cbra{\{1,2\},\{3\}}$ if and only if sub-processes $\cbra{\rbra{Y_{1,t},Y_{2,t}}}_{t>0}$ and $\cbra{\rbra{Y_{3,t}}}_{t>0}$ follow a bivariate VAR($k$) model and an AR($k$) model, respectively.

With the notation of \cref{sec:VAR}, let $\cbra{\bZ_{t}}_{t>0}$ be a Gaussian multivariate VAR($k$) process, where $\bZ_t = \rbra{Z_{1,t},\dots,Z_{d,t}}^T$, and let $\cbra{S_1,\dots,S_n}$ be a partition of $\cbra{1,\dots,d}$.
For simplicity, let $\bZ_{S_i,t}$ denote the sub-vector of $\bZ_t$ so that the dimension indices of the elements of $\bZ_{S_i,t}$ are in $S_i$, i.e. $\bZ_{S_i,t}=\rbra{Z_{s_{i,1},t},\dots, Z_{s_{i,m_i},t}}^T$ if $S_i = \cbra{s_{i,1},\dots,s_{i,m_i}}$, where the cardinality of $S_i$ is $|S_i|=m_i$.
Since $\cbra{\bZ_{t}}_{t>0}$ is a VAR process, the random vector $\rbra{\bZ_{t}^T, \dots, \bZ_{t-l}^T}^T$ follows a multivariate Gaussian distribution for any $l<t$.
As a consequence, $\rbra{\bZ_{S_i,t}^T, \dots, \bZ_{S_i,t-l}^T}^T$ is also a multivariate Gaussian random vector and we only need to ensure the process $\cbra{\bZ_{S_i,t}}_{t>0}$ has Markov order $k$ so that it follow a VAR($k$) process.
It means that $\cbra{\bZ_{S_i,t}}_{t>0}$ is a VAR($k$) process if and only if for all $t>0$, $\bZ_{S_i,t}$ is
independent of $\bZ_{S_i,t-l}$ given $\bZ_{S_i,t-1},\dots,\bZ_{S_i,t-l+1}$ for any $l$ with $k<l<t$.
We write this condition as
\[\sbra{\bZ_{S_i,t} \bot \bZ_{S_i,t-l}} | \bZ_{S_i,t-1},\dots,\bZ_{S_i,t-l+1}\ \mbox{for}\ k<l<t.\]
The following lemma is helpful in deriving a sufficient condition for $\cbra{\bZ_{S_i,t}}_{t>0}$ to have Markov order $k$.

\begin{lemma}\label{lemma:Gaussian_partial_corr}
Let $\bA$, $\bB$, $\bV$, and $\bW$ be four random vectors whose joint distribution is multivariate Gaussian.
Suppose that $\bA$ and $\bB$ are independent given $\bV$ and $\bW$.
Then $\bA$ and $\bB$ are independent given $\bV$ if $\bA$ is independent of $\bW$ given $\bV$ or $\bB$ is independent of $\bW$ given $\bV$.
\end{lemma}

\begin{proof}
  The proof is based on the result for the conditional distribution of
multivariate Gaussian random vectors.
Let $\bSigma_{\bA,\bB|\bV}$, $\bSigma_{\bA,\bW|\bV}$ and $\bSigma_{\bB,\bW|\bV}$ denote the cross-covariance matrices between  $\bA$ and $\bB$, $\bA$ and $\bW$, and $\bB$ and $\bW$, respectively, conditional on $\bV$; e.g., the $i$-th row and $j$-th column of matrix $\bSigma_{\bA,\bB|\bV}$ is the covariance between the $i$-th element of $\bA$ and $j$-th element of $\bB$ conditional on $\bV$.
Then
  \[\bSigma_{\bA,\bB|\bV,\bW} = \bSigma_{\bA,\bB|\bV} -
\bSigma_{\bA,\bW|\bV}\bSigma_{\bW|\bV}^{-1}\bSigma_{\bB,\bW|\bV}, \]
where $\bSigma_{\bA,\bB|\bV,\bW}$ is the covariance matrix between $\bA$ and
$\bB$ given $\bV$ and $\bW$, and $\bSigma_{\bW|\bV}$ is the covariance matrix of $\bW$ given $\bV$.
Since $\bA$ and $\bB$ are independent given $\bV$ and $\bW$, we have $\bSigma_{\bA,\bB|\bV,\bW}=\bm{0}$ and it follows that
  \[\bSigma_{\bA,\bB|\bV} = \bSigma_{\bA,\bW|\bV}\bSigma_{\bW|\bV}^{-1}\bSigma_{\bB,\bW|\bV},\]
and this is $\bm{0}$ if $\bSigma_{\bA,\bW|\bV}=\bm{0}$ or $\bSigma_{\bB,\bW|\bV}=\bm{0}$.
\end{proof} 
Note that the condition in \cref{lemma:Gaussian_partial_corr} is only sufficient
as $\bSigma_{\bA,\bW|\bV}\bSigma_{\bW|\bV}^{-1}\bSigma_{\bB,\bW|\bV}=\bm{0}$ does not lead to either
$\bSigma_{\bA,\bW|\bV}=\bm{0}$ or $\bSigma_{\bB,\bW|\bV}=\bm{0}$.

Since $\cbra{\bZ_t}_{t>0}$ is a VAR($k$) process and hence Markov of order $k$, $\bZ_{S_i,t}$ is independent of $\bZ_{S_i,t-l}$ conditional on $\rbra{\bZ_{t-1}^T,\dots,\bZ_{t-l+1}^T}^T$ for any $l$ with $k<l<t$. Let $\bZ_{-S_i,t}$ denote the sub-vector of $\bZ_t$ with components in index set $\cbra{1,\dots,d}\backslash S_i$ and set
\[
  \bA = \bZ_{S_i,t}, & \quad \bB = \bZ_{S_i,t-l}, \\
  \bV = \rbra{\bZ_{S_i,t-1}^T,\dots,\bZ_{S_i,t-l+1}^T}^T, & \quad \bW = \rbra{\bZ_{-S_i,t-1}^T,\dots,\bZ_{-S_i,t-l+1}^T}^T.
\]
Then using \cref{lemma:Gaussian_partial_corr}, we have $\sbra{\bZ_{S_i,t} \bot \bZ_{S_i,t-l}} | \bZ_{S_i,t-1},\dots,\bZ_{S_i,t-l+1}$ if
\begin{align}\label{eq:condition_of_being_mc}
  \sbra{\bZ_{S_i,t} \bot \rbra{\bZ_{-S_i,t-1}^T,\dots,\bZ_{-S_i,t-l+1}^T}^T} \big| \bZ_{S_i,t-1},\dots,\bZ_{S_i,t-l+1}\ \mbox{or}\\
  \sbra{\bZ_{S_i,t-l} \bot \rbra{\bZ_{-S_i,t-1}^T,\dots,\bZ_{-S_i,t-l+1}^T}^T}\big | \bZ_{S_i,t-1},\dots,\bZ_{S_i,t-l+1}.
\end{align}
The random vectors $\bA$ and $\bB$ are time lag $l$ apart (for sub-vector with indices in $S_i$),
$\bV$ consists of the intermediate observations in this sub-process,
and $\bW$ consists of the intermediate observations in the other sub-process.
It follows that $\cbra{\bZ_{S_i,t}}_{t>0}$ is a VAR($k$) process if the condition in \cref{eq:condition_of_being_mc} holds for any $l$ with $k<l<t$.
Moreover, it can be proved that we only need to guarantee the condition at $l=k+1$ in order to make $\cbra{\bZ_{S_i,t}}_{t>0}$ follow a VAR($k$) process.
This is formulated in the next theorem, with proof given in the Appendix.

\begin{theorem}\label{thm:sufficient_condition}
If $\cbra{\bZ_{t}}_{t>0}$ is a multivariate VAR($k$) process where $\bZ_t = \rbra{Z_{1,t},\dots,Z_{d,t}}^T$, and $\cbra{S_1,\dots,S_n}$ is a partition of $\cbra{1,\dots,d}$.
Let  $-S_i$ denote the difference between two sets $\cbra{1,\dots,d}$ and
$S_i$, and let $\bZ_{S_i,t}$ denote the sub-vector of $\bZ_t$ with components in index set $S_i$.
Then the sub-process $\cbra{\bZ_{S_i,t}}_{t>0}$ follows a VAR($k$) model if one of the following two conditions is satisfied:
\begin{enumerate}
  \item $\sbra{\bZ_{S_i,t} \bot \rbra{\bZ_{-S_i,t-1}^T,\dots,\bZ_{-S_i,t-k}^T}^T} \big| \bZ_{S_i,t-1},\dots,\bZ_{S_i,t-k}$,
  \item $\sbra{\bZ_{S_i,t-k-1} \bot \rbra{\bZ_{-S_i,t-1}^T,\dots,\bZ_{-S_i,t-k}^T}^T} \big| \bZ_{S_i,t-1},\dots,\bZ_{S_i,t-k}$.
\end{enumerate}
\end{theorem}

For illustration, consider the VAR(1) model in \cref{eq:example_of_VAR(1)}.
When
$a_{12}=a_{21}=0$, $-1<a_{11}<1$, and $-1<a_{22}<1$,
both $\cbra{Y_{1,t}}_{t>0}$ and $\cbra{Y_{2,t}}_{t>0}$ follow AR(1) models and the VAR(1) process in \cref{eq:example_of_VAR(1)} is closed under margins with respect to partition $\{\{1\},\{2\}\}$.
It is easy to verify that $\sbra{Y_{1,t} \bot
Y_{2,t-1}}|Y_{1,t-1}$ and $\sbra{Y_{2,t} \bot Y_{1,t-1}}|Y_{2,t-1}$ in this
case; these correspond to the first conditions in \cref{thm:sufficient_condition} with $S_i = \{1\}$ and $S_i = \{2\}$, respectively.

The condition in \cref{thm:sufficient_condition} is only sufficient because \cref{lemma:Gaussian_partial_corr} specifies a sufficient condition. A necessary and sufficient condition for the closure under margins would require investigating a necessary and sufficient condition under which $\bSigma_{\bA,\bW|\bV}\bSigma_{\bW|\bV}^{-1}\bSigma_{\bB,\bW|\bV}=\bm{0}$. We do not see a simple formulation for such a condition at the moment and hence the question is left for future research. However, we do note that sufficient conditions are adequate for model fitting.


\subsection{Conditions on parameters for closure under margins}
\label{sec:closure2}

In this section, conditions for conditional independence in \cref{thm:sufficient_condition}
are expressed in terms of model parameters.

\subsubsection{Partitions with two sub-processes}
\label{sec:partition_with_two}

We consider the case with partition $\cbra{S_1, S_2}$ and there are only two sub-processes $\cbra{\bZ_{S_1,t}}_{t>0}$ and $\cbra{\bZ_{S_2,t}}_{t>0}$.
Let the dimensions of the sub-vectors be $d_1$ and $d_2$, respectively, and let $R$ be the correlation matrix of the random vector $\rbra{\bZ_t^T,\dots,\bZ_{t-k}^T}^T$.
To better present the structure of the sub-processes, we reorder the columns and rows of $R$ to get the correlation matrix of the random vector $\rbra{\bZ_{S_1,t}^T, \dots, \bZ_{S_1,t-k}^T, \bZ_{S_2,t}^T, \dots, \bZ_{S_2,t-k}^T}^T$.
Denote the reordered matrix by $R_{\{S_1,S_2\}}$, with notation:
\begin{equation}\label{eq:RS1S2_of_VAR}
\begin{aligned}
  R_{\{S_1,S_2\}} &=
\begin{pmatrix}
\begin{BMAT}{c:c}{c:c}
  R_{S_1} & R_{S_1,S_2} \\
  R^T_{S_1,S_2} & R_{S_2}
\end{BMAT}
\end{pmatrix} \\
&=
\begin{pmatrix}
 \begin{BMAT}{cccc:cccc}{cccc:cccc}
  \bSigma_{11,0}   & \bSigma_{11,1}     & \cdots & \bSigma_{11,k}   & \bSigma_{12,0}   & \bSigma_{12,1}     & \cdots & \bSigma_{12,k}\\
  \bSigma_{11,1}^T & \bSigma_{11,0}     & \cdots & \bSigma_{11,k-1} & \bSigma_{12,-1}  & \bSigma_{12,0}     & \cdots & \bSigma_{12,k-1}\\
  \vdots           & \vdots             & \ddots & \vdots           & \vdots           & \vdots             & \ddots & \vdots\\
  \bSigma_{11,k}^T & \bSigma_{11,k-1}^T & \cdots & \bSigma_{11,0}   & \bSigma_{12,-k}  & \bSigma_{12,1-k}   & \cdots & \bSigma_{12,0}\\
  \bSigma_{21,0}   & \bSigma_{21,1}     & \cdots & \bSigma_{21,k}   & \bSigma_{22,0}   & \bSigma_{22,1}     & \cdots & \bSigma_{22,k} \\
  \bSigma_{21,-1}  & \bSigma_{21,0}     & \cdots & \bSigma_{21,k-1} & \bSigma_{22,1}^T & \bSigma_{22,0}     & \cdots & \bSigma_{22,k-1} \\
  \vdots           & \vdots             & \ddots & \vdots           & \vdots           & \vdots             & \ddots & \vdots \\
  \bSigma_{21,-k}  & \bSigma_{21,1-k}   & \cdots & \bSigma_{21,0}   & \bSigma_{22,k}^T & \bSigma_{22,k-1}^T & \cdots & \bSigma_{22,0}
 \end{BMAT}
 \end{pmatrix},
\end{aligned}
\end{equation}
where the block entry $\bSigma_{ij,l}$ is the correlation matrix between
$\bZ_{i,t}$ and $\bZ_{j,t-l}$ for $i,j\in \{1,2\}$ and $l=0,1,\dots,k$.
These are treated as parameters of $R_{\{S_1,S_2\}}$.
Under the constraint that $\cbra{\bZ_{t}}_{t>0}$ is closed under margins with respect to partition $\{S_1,S_2\}$, the sub-matrix $R_{S_1}$, which is the correlation matrix of $\rbra{\bZ_{S_1,t}^T,\dots,\bZ_{S_1,t-k}^T}^T$,
specifies the VAR($k$) sub-process $\cbra{\bZ_{S_1,t}}_{t>0}$.
More specifically, $\bSigma_{11,0}$ specifies the contemporaneous dependence and $\bSigma_{11,1},\dots, \bSigma_{11,k}$ specify the serial dependence at different lags of $\cbra{\bZ_{S_1,t}}_{t>0}$.
Similarly, the sub-matrix $R_{S_2}$ characterizes the second VAR($k$) sub-process $\cbra{\bZ_{S_2,t}}_{t>0}$.
As for the other block entries in the off-diagonal block $R_{S_1,S_2}$,
they measure the dependence between two sub-processes: $\bSigma_{12,0}$
indicates the contemporaneous dependence between $\bZ_{S_1,t}$ and
$\bZ_{S_2,t}$, while $\bSigma_{12,-k},\dots, \bSigma_{12,-1},
\bSigma_{12,1},\dots,\bSigma_{12,-k}$ summarize the cross-sectional dependence between two sub-processes at lags $-k$ to $k$.
It is natural to consider closure under margins as a property describing the dependence structure between the sub-processes.
Therefore, when imposing closure under margins for $\cbra{\bZ_{t}}_{t>0}$ with respect to the
partition including the two sub-processes, we hold the entries modeling the
individual VAR($k$) sub-processes fixed and investigate the constraints on
the entries for the dependence between the sub-processes, i.e., the
constraints on parameters $\bSigma_{12,-k},\dots,\bSigma_{12,k}$ with
parameters $\bSigma_{11,0},\bSigma_{22,0},\dots, \bSigma_{11,k},\bSigma_{22,k}$ fixed.

To reformulate the conditions in \cref{thm:sufficient_condition} as the constraints on entries of $R_{\{S_1,S_2\}}$ in \cref{eq:RS1S2_of_VAR}, we first point out that $S_1=-S_2$ and $-S_1=S_2$ in this case of only two sub-processes.
Then the conditions in \cref{thm:sufficient_condition} can be expressed as
\begin{equation}\label{eq:thm3_for_S1S2}
\begin{aligned}
  \mbox{1}. &\sbra{\bZ_{S_1,t} \bot \rbra{\bZ_{S_2,t-1}^T,\dots,\bZ_{S_2,t-k}^T}^T} \big| \bZ_{S_1,t-1},\dots,\bZ_{S_1,t-k},\\
  \mbox{2}. &\sbra{\bZ_{S_1,t-k-1} \bot \rbra{\bZ_{S_2,t-1}^T,\dots,\bZ_{S_2,t-k}^T}^T} \big| \bZ_{S_1,t-1},\dots,\bZ_{S_1,t-k}.
\end{aligned}
\end{equation}

The reader may want to look ahead at Examples
\ref{ex:2_dimensional_VAR(1)} and
\ref{ex:3_dimensional_VAR(1)} to get an idea of what is involved when
the dependence parameters of the two sub-processes are fixed
and some cross-dependence parameters are fixed, in order to
construct $d$-dimensional VAR($k$) processes that have two
marginal VAR($k$) sub-processes.

\medskip

Let $$D_1 = \rbra{ \bSigma_{12, -k}^T, \dots, \bSigma_{12,-1}^T, \bSigma_{12,0}^T,
\bSigma_{12,1}^T, \dots, \bSigma_{12,k}^T }^T,$$
\begin{align}
\begin{pmatrix}
   \bPhi_{1,1}^T \\
   \bPhi_{1,2}^T \\
   \vdots \\
   \bPhi_{1,k}^T
 \end{pmatrix}^T &=
 \begin{pmatrix}
  \bSigma_{11,1}^T \\
  \bSigma_{11,2}^T \\
  \vdots \\
  \bSigma_{11,k}^T
 \end{pmatrix}^T
 \begin{pmatrix}
  \bSigma_{11,0}     & \bSigma_{11,1}     & \cdots & \bSigma_{11,k-1} \\
  \bSigma_{11,1}^T   & \bSigma_{11,0}     & \cdots & \bSigma_{11,k-2} \\
  \vdots             & \vdots             & \ddots & \vdots \\
  \bSigma_{11,k-1}^T & \bSigma_{11,k-2}^T & \cdots & \bSigma_{11,0}
 \end{pmatrix}^{-1},\\\\
 \begin{pmatrix}
   \bPsi_{1,1}^T \\
   \bPsi_{1,2}^T \\
   \vdots \\
   \bPsi_{1,k}^T
 \end{pmatrix}^T &=
 \begin{pmatrix}
  \bSigma_{11,-k}^T \\
  \bSigma_{11,1-k}^T \\
  \vdots \\
  \bSigma_{11,-1}^T
 \end{pmatrix}^T
 \begin{pmatrix}
  \bSigma_{11,0}     & \bSigma_{11,1}     & \cdots & \bSigma_{11,k-1} \\
  \bSigma_{11,1}^T   & \bSigma_{11,0}     & \cdots & \bSigma_{11,k-2} \\
  \vdots             & \vdots             & \ddots & \vdots \\
  \bSigma_{11,k-1}^T & \bSigma_{11,k-2}^T & \cdots & \bSigma_{11,0}
 \end{pmatrix}^{-1},
\end{align}
\begin{equation}\label{eq:G1}
  G_1 =
   \begin{pmatrix}
  0      & \bPhi_{1,k} & \dots       & \bPhi_{1,2} & \bPhi_{1,1} & -I_{d_1}    & 0          & \cdots      & 0\\
  0      & 0           & \bPhi_{1,k} & \dots       & \bPhi_{1,2} & \bPhi_{1,1} & -I_{d_1}  & \cdots      & 0 \\
  \vdots & \ddots      & \ddots      & \ddots      & \ddots      & \ddots      &   \ddots    & \ddots      & \vdots \\
  0      & \cdots      & 0           & 0           & \bPhi_{1,k} & \cdots      & \bPhi_{1,2} & \bPhi_{1,1} & -I_{d_1}
 \end{pmatrix}
\end{equation}
and
\begin{equation}\label{eq:H1}
  H_1 =
  \begin{pmatrix}
  -I_{d_1} & \bPsi_{1,k} & \dots       & \bPsi_{1,2} & \bPsi_{1,1} & 0           & 0           & \cdots      & 0\\
  0          & -I_{d_1}  & \bPsi_{1,k} & \dots       & \bPsi_{1,2} & \bPsi_{1,1} & 0           & \cdots      & 0 \\
  \vdots     & \ddots      & \ddots      & \ddots      & \ddots      & \ddots      &   \ddots    & \ddots      & \vdots \\
  0          & \cdots      & 0           & -I_{d_1}  & \bPsi_{1,k} & \cdots      & \bPsi_{1,2} & \bPsi_{1,1} & 0
 \end{pmatrix}.
\end{equation}
Then the two conditions in \cref{eq:thm3_for_S1S2}
for $\cbra{\bZ_{S_1,t}}_{t>0}$ can be rewritten as two linear systems:
\begin{equation}\label{eq: two_linear_system1}
 \begin{aligned}
\mbox{Condition 1: } &
 G_1D_1 = \bm{0},\\
\mbox{Condition 2: } &
 H_1D_1 = \bm{0}.
\end{aligned}
\end{equation}
Note that block coefficient matrices $G_1$ and $H_1$, defined in
\cref{eq:G1} and \cref{eq:H1},
have $k\times (2k+1)$ blocks of dimension $d_1\times d_1$.
Matrix $D_1$ has $(2k+1)\times 1$ blocks of dimension $d_1\times d_2$.
Moreover, one can see that $\bPhi_{1,1},\dots,\bPhi_{1,k}$ are the coefficient matrices of the VAR($k$) process
$\cbra{\bZ_{S_1,t}}_{t>0}$ when $\cbra{\bZ_{t}}_{t>0}$ is closed under margins
with respect to partition $\{S_1,S_2\}$,
and so these coefficient matrices come from the conditional mean of $\bZ_{S_1,t}$ given $\bZ_{S_1,t-1},\dots,\bZ_{S_1,t-k}$ written as $\bPhi_{1,1}\bZ_{S_1,t-1}+\cdots+\bPhi_{1,k}\bZ_{S_1,t-k}$.
The $\bPsi_{1,j}$ coefficient matrices come from the conditional
mean of $\bZ_{S_1,t-k-1}$ given $\bZ_{S_1,t-1},\dots,\bZ_{S_1,t-k}$.

Similarly, let
\begin{multline}
  D_2 = \rbra{ \bSigma_{21, -k}^T, \dots, \bSigma_{21,-1}^T, \bSigma_{21,0}^T, \bSigma_{21,1}^T, \dots, \bSigma_{21,k}^T }^T = \\\rbra{ \bSigma_{12, k}, \dots, \bSigma_{12,1}, \bSigma_{12,0}, \bSigma_{12,-1}, \dots, \bSigma_{12,-k} }^T,
\end{multline}
\begin{align}
\begin{pmatrix}
  \bPhi_{2,1}^T \\
  \bPhi_{2,2}^T \\
  \vdots \\
  \bPhi_{2,k}^T
 \end{pmatrix}^T &=
 \begin{pmatrix}
  \bSigma_{22,1}^T \\
  \bSigma_{22,2}^T \\
  \vdots \\
  \bSigma_{22,k}^T
 \end{pmatrix}^T
 \begin{pmatrix}
  \bSigma_{22,0}     & \bSigma_{22,1}     & \cdots & \bSigma_{22,k-1} \\
  \bSigma_{22,1}^T   & \bSigma_{22,0}     & \cdots & \bSigma_{22,k-2} \\
  \vdots             & \vdots             & \ddots & \vdots \\
  \bSigma_{22,k-1}^T & \bSigma_{22,k-2}^T & \cdots & \bSigma_{22,0}
 \end{pmatrix}^{-1},\\\\
 \begin{pmatrix}
   \bPsi_{2,1}^T \\
   \bPsi_{2,2}^T \\
   \vdots \\
   \bPsi_{2,k}^T
 \end{pmatrix}^T &=
 \begin{pmatrix}
  \bSigma_{22,-k}^T \\
  \bSigma_{22,1-k}^T \\
  \vdots \\
  \bSigma_{22,-1}^T
 \end{pmatrix}^T
 \begin{pmatrix}
  \bSigma_{22,0}     & \bSigma_{22,1}     & \cdots & \bSigma_{22,k-1} \\
  \bSigma_{22,1}^T   & \bSigma_{22,0}     & \cdots & \bSigma_{22,k-2} \\
  \vdots             & \vdots             & \ddots & \vdots \\
  \bSigma_{22,k-1}^T & \bSigma_{22,k-2}^T & \cdots & \bSigma_{22,0}
 \end{pmatrix}^{-1},
\end{align}
and
\begin{equation}\label{eq:G2}
  G_2 =
  \begin{pmatrix}
  0      & \bPhi_{2,k} & \dots       & \bPhi_{2,2} & \bPhi_{2,1} & -I_{d_2}  & 0           & \cdots      & 0\\
  0      & 0           & \bPhi_{2,k} & \dots       & \bPhi_{2,2} & \bPhi_{2,1} & -I_{d_2}  & \cdots      & 0 \\
  \vdots & \ddots      & \ddots      & \ddots      & \ddots      & \ddots      &   \ddots    & \ddots      & \vdots \\
  0      & \cdots      & 0           & 0           & \bPhi_{2,k} & \cdots      & \bPhi_{2,2} & \bPhi_{2,1} & -I_{d_2}
 \end{pmatrix},
\end{equation}
\begin{equation}\label{eq:H2}
  H_2 =
  \begin{pmatrix}
  -I_{d_2} & \bPsi_{2,k} & \dots       & \bPsi_{2,2} & \bPsi_{2,1} & 0           & 0           & \cdots      & 0\\
  0          & -I_{d_2}  & \bPsi_{2,k} & \dots       & \bPsi_{2,2} & \bPsi_{2,1} & 0           & \cdots      & 0 \\
  \vdots     & \ddots      & \ddots      & \ddots      & \ddots      & \ddots      &   \ddots    & \ddots      & \vdots \\
  0          & \cdots      & 0           & -I_{d_2}  & \bPsi_{2,k} & \cdots      & \bPsi_{2,2} & \bPsi_{2,1} & 0
 \end{pmatrix}.
\end{equation}
Define $L_2 = J_{2k+1} \otimes I_{d_2}$, where $J_{2k+1}$ is the $(2k+1)$-dimensional exchange matrix whose
elements in the anti-diagonal (or opposite diagonal) are one and all other elements are zero, and $\otimes$ denotes the Kronecker product.
Then $L_2D_{2} = L_2^{-1}D_{2} = \rbra{ \bSigma_{12, -k}, \dots, \bSigma_{12,-1}, \bSigma_{12,0}, \bSigma_{12,1}, \dots, \bSigma_{12,k} }^T.$
It can be seen that the blocks in $D_1$ and $L_2D_{2}$ are stacked in the same order of indices
while the blocks of $L_2D_{2}$ are transposes of the blocks in $D_1$.
Conditions 1 and 2 for $\cbra{\bZ_{S_2,t}}_{t>0}$ can be rewritten as
the following two linear systems:
\begin{equation}\label{eq: two_linear_system2}
 \begin{aligned}
\mbox{Condition 1: } &
 (G_2L_2)(L_2D_2) = \bm{0},\\
\mbox{Condition 2: } &
 (H_2L_2)(L_2D_2) = \bm{0}.
\end{aligned}
\end{equation}
The block coefficient matrices $G_2$ and $H_2$, defined in
\cref{eq:G2} and \cref{eq:H2},
have $k\times (2k+1)$ blocks of dimension $d_2\times d_2$,
and $D_2$ has $(2k+1)\times 1$ blocks of dimension $d_2\times d_1$.
Note that $\bPhi_{2,1},\dots,\bPhi_{2,k}$ are the coefficient matrices of the VAR($k$) process $\cbra{\bZ_{S_2,t}}_{t>0}$ when $\cbra{\bZ_{t}}_{t>0}$ is closed under margins with respect to partition $\{S_1,S_2\}$.
Also , $\bPsi_{2,1},\dots,\bPsi_{2,k}$ come from the conditional
mean of $\bZ_{S_2,t-k-1}$ given  $\bZ_{S_2,t-1},\ldots,\bZ_{S_2,t-k}$.
The derivations of \cref{eq: two_linear_system1} and \cref{eq: two_linear_system2}
are given in \cref{appd: derive_linear_systems}.

To make the model closed under margins with respect to partition $\cbra{S_1, S_2}$,
we only need to take one system from
\cref{eq: two_linear_system1},
and another system from
\cref{eq: two_linear_system2},
so there are 4 different combinations to consider.
Moreover,
$G_1, G_2, H_1, H_2$ have $2k+1$ blocks in each row and $k$ blocks in each column, so for each combination we have only $2k$ linear equations for $2k+1$ block matrices $\bSigma_{12, -k}, \dots, \bSigma_{12,-1}, \bSigma_{12,0}, \bSigma_{12,1}, \dots, \bSigma_{12,k}$.
Thus, we must choose a block matrix amongst them
and fix it when solving a system of $2k$ linear equations
in $2k$ block matrices of dimensions $d_1\times d_2$.

It can be concluded that when parameters $\bSigma_{11,0},\bSigma_{22,0},\dots,\bSigma_{11,k}, \bSigma_{22,k}$ modeling two individual VAR($k$) sub-processes $\cbra{\bZ_{S_1,t}}_{t>0}$ and $\cbra{\bZ_{S_2,t}}_{t>0}$ are fixed and the chosen block matrix is given, all parameters modeling cross-sectional dependence between the two sub-processes can be uniquely solved,
leading to a correlation matrix $R_{\{S_1,S_2\}}$ for $\rbra{\bZ_{t}^T,\dots,\bZ_{t-k}^T}^T$.
If this $R_{\{S_1,S_2\}}$ is positive definite, then the
parameters of the two VAR($k$) subprocesses and the chosen block matrix
are compatible for a margin-closed $d$-dimensional VAR($k$) process.

Next we consider four cases based on conditions in
\cref{eq: two_linear_system1} and \cref{eq: two_linear_system2}
and give two examples to illustrate the details:

\medskip
\noindent\textbf{Case 1}. $\cbra{\bZ_{S_1,t}}_{t>0}$ and $\cbra{\bZ_{S_2,t}}_{t>0}$ satisfy  Condition 1 in
\cref{eq: two_linear_system1} and \cref{eq: two_linear_system2}
respectively;

\noindent\textbf{Case 2}. $\cbra{\bZ_{S_1,t}}_{t>0}$ and $\cbra{\bZ_{S_2,t}}_{t>0}$ satisfy Condition 2 in
\cref{eq: two_linear_system1} and \cref{eq: two_linear_system2}
respectively;

\noindent\textbf{Case 3}. $\cbra{\bZ_{S_1,t}}_{t>0}$ satisfies Condition 1 in
\cref{eq: two_linear_system1}
and $\cbra{\bZ_{S_2,t}}_{t>0}$ satisfies Condition 2 in
\cref{eq: two_linear_system2};

\noindent\textbf{Case 4}. $\cbra{\bZ_{S_1,t}}_{t>0}$ satisfies Condition 2 in
\cref{eq: two_linear_system1}
and $\cbra{\bZ_{S_2,t}}_{t>0}$ satisfies Condition 1 in
\cref{eq: two_linear_system2}.

\medskip

In Cases 1 and 2, it can be seen that when the parameters of two individual VAR($k$) sub-processes are held fixed, the dependence between $\cbra{\bZ_{S_1,t}}_{t>0}$ and $\cbra{\bZ_{S_2,t}}_{t>0}$ can also be specified when their contemporaneous dependence parameter $\bSigma_{12,0}$ is fixed.
The cross-sectional dependence parameters $\bSigma_{12, -k}, \dots,
\bSigma_{12,-1}$, $\bSigma_{12,1}, \dots, \bSigma_{12,k}$ can be uniquely determined.

Things will be different in Case 3.
Since all block entries in the first column of $G_1$ and $H_2L_2$ are $\bm{0}$, $\bSigma_{12,-k}$ can not be solved and should be selected as the fixed parameter.
Moreover, we can see that $\bSigma_{12,1-k}=\cdots=\bSigma_{12,k}=\bm{0}$ as there are non-zero entries in each of the last $2k$ block columns of $G_1$ and $H_2L_2$.
It follows that when the parameters of two individual VAR($k$) sub-processes $\cbra{\bZ_{S_1,t}}_{t>0}$ and $\cbra{\bZ_{S_2,t}}_{t>0}$ are held fixed, the cross-sectional dependence between the two sub-processes at lag $-k$ is the fixed parameter based on which the dependence between the two sub-processes can be determined.

In Case 4,
$\bSigma_{12,-k},\dots,\bSigma_{12,k-1}$ are all $\bm{0}$ while $\bSigma_{12,k}$, which is the cross-sectional dependence between $\cbra{\bZ_{S_1,t}}_{t>0}$ and $\cbra{\bZ_{S_2,t}}_{t>0}$ at lag $k$, should be treated as the fixed parameter after fixing the parameters of two individual VAR($k$) sub-processes.

To summarize, the correlation matrix $R_{\{S_1,S_2\}}$ can be specified by three parts:
the conditions that two sub-processes take, the entries of $R_{S_1}$ and $R_{S_2}$,
and the fixed cross-dependence parameter according to the conditions of two sub-processes.
Note that the fixed parameter is $\bSigma_{12, 0}$ in Cases 1 and 2,
is $\bSigma_{12, -k}$ in Case 3,
and is $\bSigma_{12, k}$ in Case 4.
The extra constraint that $R_{\{S_1,S_2\}}$ is positive definite should
hold so that the fixed parameters are compatible
for a margin-closed $d$-dimensional VAR($k$) process.

The next two examples illustrate four cases above for partitions with two sub-processes.
\cref{ex:2_dimensional_VAR(1)} is the extension of the bivariate VAR(1) model in \cref{eq:example_of_VAR(1)}.
\cref{ex:3_dimensional_VAR(1)} is more complex as it deals with a bivariate sub-process.
General formulas for solving the linear equations of Case 1 and 2 are derived in \cref{appd:two_sub_processes}.

\begin{example}\label{ex:2_dimensional_VAR(1)}
(2-dimensional VAR(1) process).
Consider a 2-dimensional VAR(1) model $\{(Z_{1,t},Z_{2,t})\}_{t>0}$
with simplest partition $\cbra{\{1\}, \{2\}}$.
Following the same assumption in \cref{sec:VAR},
$Z_{i,t}$ has standard Gaussian margin for all $i\in\cbra{1,2}$ and $t>0$.
Let $\rho_{ij,l}$ denote the correlation coefficient between $Z_{i,t}$ and $Z_{j,t-l}$ for $i,j\in\cbra{1,2}$.
The stationary joint distribution of $\cbra{Z_{1,t}, Z_{1,t-1}, Z_{2,t}, Z_{2,t-1}}^T$
has correlation matrix
\begin{equation} \label{eq:bivariateVAR1cor}
R_{\cbra{\{1\}, \{2\}}} =
\begin{pmatrix}
\begin{BMAT}{c:c}{c:c}
  R_{\{1\}}         & R_{\{1\},\{2\}} \\
  R^T_{\{1\},\{2\}} & R_{\{2\}}
\end{BMAT}
\end{pmatrix}
=
\begin{pmatrix}
\begin{BMAT}{cc:cc}{cc:cc}
  1           & \rho_{11,1} & \rho_{12,0} & \rho_{12,1} \\
  \rho_{11,1} & 1           & \rho_{12,-1}& \rho_{12,0} \\
  \rho_{12,0} & \rho_{12,-1}& 1           & \rho_{22,1} \\
  \rho_{12,1} & \rho_{12,0} & \rho_{22,1} & 1
\end{BMAT}
\end{pmatrix}.
\end{equation}
If $R_{\cbra{\{1\}, \{2\}}}$ in \cref{eq:bivariateVAR1cor} is positive definite,
then from the conditional expectation formula for multivariate Gaussian distributions,
the coefficient matrix of the 2-dimensional VAR(1) process above is
\begin{equation} \label{eq:bivariateVAR1Phi}
  \bPhi =
 \begin{pmatrix}
      \rho_{11,1}  & \rho_{12,1} \\
      \rho_{12,-1} & \rho_{22,1}
    \end{pmatrix}
    \begin{pmatrix}
      1 & \rho_{12,0} \\
      \rho_{12,0} & 1
    \end{pmatrix}^{-1}=
 \begin{pmatrix}
      \frac{\rho_{11,1} - \rho_{12,1} \rho_{12,0}}{1-\rho_{12,0}^2} &
      \frac{\rho_{12,1} -
      \rho_{11,1} \rho_{12,0}}{1-\rho_{12,0}^2} \\
      \frac{\rho_{12,-1} - \rho_{22,1} \rho_{12,0}}{1-\rho_{12,0}^2} &
      \frac{\rho_{22,1} -
      \rho_{12,-1} \rho_{12,0} }{1-\rho_{12,0}^2} \\
    \end{pmatrix}.
\end{equation}

If $\{(Z_{1,t},Z_{2,t})\}_{t>0}$ is closed under margins with respect to partition $\cbra{\{1\}, \{2\}}$, the entries of sub-matrix $R_{\{1\}}$ can characterize the VAR(1) sub-process $\cbra{Z_{1,t}}_{t>0}$.
Specifically, we should have
\[ Z_{1,t} = \rho_{11,1}Z_{1,t-1} + \epsilon'_{1,t},\quad \epsilon'_{1,t}\overset{\mathbf{i.i.d.}}{\sim} \mathcal{N}(0, 1-\rho_{11,1}^2). \]
Similarly, the VAR(1) sub-process $\cbra{Z_{1,t}}_{t>0}$ is specified by the entries of $R_{\{2\}}$:
\[ Z_{2,t} = \rho_{22,1}Z_{2,t-1} + \epsilon'_{2,t},\quad \epsilon'_{2,t}\overset{\mathbf{i.i.d.}}{\sim} \mathcal{N}(0, 1-\rho_{22,1}^2). \]
Then $\rho_{11,1}$ and $\rho_{22,1}$ are considered fixed and
one cross-correlation of lag $0$, $1$ or $-1$ must be fixed in order
for $\cbra{\bZ_{t}}_{t>0}$ to satisfy the margin-closure property.
As $\bPhi_{1,1}=\bPsi_{1,1}=\rbra{\rho_{11,1}}$ and
$\bPhi_{2,1}=\bPsi_{2,1}=\rbra{\rho_{22,1}}$,
in
\cref{eq:G1} to \cref{eq:H2},
the equations involving
$G_1,H_1,G_2,H_2$ all involve scalars:
\begin{equation}
 \begin{aligned}
 G_1D_1=0 \Rightarrow &\ \bPhi_{1,1}\rho_{12,0}-\rho_{12,1}=\rho_{11,1}\rho_{12,0}-\rho_{12,1}=0,\\
 H_1D_1=0 \Rightarrow &\ -\rho_{12,-1}+\bPsi_{1,1}\rho_{12,0}=-\rho_{12,-1}+\rho_{11,1}\rho_{12,0}=0,\\
 (G_2L_2)(L_2D_2)=0 \Rightarrow &\ -\rho_{12,-1}+\bPhi_{2,1}\rho_{12,0}=-\rho_{12,-1}+\rho_{22,1}\rho_{12,0}=0,\\
 (H_2L_2)(L_2D_2)=0 \Rightarrow &\ \bPsi_{2,1}\rho_{12,0}-\rho_{12,1}=\rho_{22,1}\rho_{12,0}-\rho_{12,1}=0.
 \end{aligned}
\end{equation}

In Case 1, based on $G_1$ and $G_2$,
\[\rho_{12,-1} =\rho_{22,1}\rho_{12,0}, \quad
  \rho_{12,1} = \rho_{11,1}\rho_{12,0}.\]

If $\rho_{12,0}$ is such that $R_{\cbra{\{1\}, \{2\}}}$
is positive definite, the coefficient matrix in
\cref{eq:bivariateVAR1Phi} is
 $\bPhi= \rbra{\begin{smallmatrix}
      \rho_{11,1} & 0 \\
      0 & \rho_{22,1}
    \end{smallmatrix}}$ and this is diagonal,
which means it corresponds to the case of a diagonal coefficient matrix of $\cbra{\bZ_t}_{t>0}$.

In Case 2 based on $H_1$ and $H_2$,
the contemporaneous dependence $\rho_{12,0}$ is the fixed parameter, and
\begin{equation} \label{eq:bivariateVAR1case2}
  \rho_{12,-1} =\rho_{11,1}\rho_{12,0}, \quad
  \rho_{12,1} = \rho_{22,1}\rho_{12,0}.
\end{equation}
If $\rho_{11,1},\rho_{22,1},\rho_{12,0}$ in \cref{eq:bivariateVAR1case2}
are such that $R_{\cbra{\{1\}, \{2\}}}$ in \cref{eq:bivariateVAR1cor} is positive definite,
the coefficient matrix in \cref{eq:bivariateVAR1Phi} is
 $$
  \bPhi=  (1-\rho_{12,0}^2)^{-1}  \begin{pmatrix}
     \rho_{11,1}-\rho_{12,0}^2\rho_{22,1} & \rho_{12,0}(\rho_{22,1}-\rho_{11,1}) \\
     \rho_{12,0}(\rho_{11,1}-\rho_{22,1}) & \rho_{22,1}-\rho_{12,0}^2\rho_{11,1}
    \end{pmatrix}.
 $$
This is the most interesting case as this coefficient matrix is non-diagonal
if $\rho_{12,0}\neq0$ and $\rho_{11,1}\neq\rho_{22,1}$.
Even though all entries of the coefficient matrix are non-zero,
the two univariate sub-processes are AR(1).

In Case 3 based on $G_1$ and $H_2$, $\rho_{12,0}=\rho_{12,1}=0$
if $\rho_{11,1}\ne\rho_{22,1}$,
If $\rho_{11,1},\rho_{22,1},\rho_{12,-1}$ with  $\rho_{12,0}=\rho_{12,1}=0$
are such that $R_{\cbra{\{1\}, \{2\}}}$ in \cref{eq:bivariateVAR1cor} is positive definite,
the coefficient matrix in \cref{eq:bivariateVAR1Phi} is
$\rbra{
\begin{smallmatrix}
 \rho_{11,1}  & 0 \\
 \rho_{12,-1} & \rho_{22,1}
\end{smallmatrix}}$.
This is non-diagonal in general, but clearly $\{Z_{1,t}\}_{t>0}$ is AR(1).

In Case 4 based on $G_2$ and $H_1$, $\rho_{12,0}=\rho_{12,-1}=0$
if $\rho_{11,1}\ne\rho_{22,1}$,
If $\rho_{11,1},\rho_{22,1},\rho_{12,1}$ with  $\rho_{12,0}=\rho_{12,-1}=0$
are such that $R_{\cbra{\{1\}, \{2\}}}$ in \cref{eq:bivariateVAR1cor} is positive definite,
the coefficient matrix in \cref{eq:bivariateVAR1Phi} is
$\rbra{
\begin{smallmatrix}
 \rho_{11,1} & \rho_{12,1} \\
 0 & \rho_{22,1}
\end{smallmatrix}}$.
This is non-diagonal in general, but clearly $\{Z_{2,t}\}_{t>0}$ is AR(1).

Note that in Case 2, the range of $\rho_{12,0}$ making $R_{\cbra{\{1\}, \{2\}}}$ in \cref{eq:bivariateVAR1cor}
positive definite is affected by the values of $\rho_{11,1}$ and
$\rho_{22,1}$.
To illustrate, numerical calculations show that when $\rho_{11,1}=\rho_{22,1}=0.9$,
$R_{\cbra{\{1\}, \{2\}}}$ is positive definite for any $\rho_{12,0}\in(-1, 1)$;
but when $\rho_{11,1}=0.9$ and $\rho_{22,1}=-0.9$, $R_{\cbra{\{1\}, \{2\}}}$
is not positive definite for $\rho_{12,0}\in(0.15, 1)$.
\end{example}

For examples where $d_1>1$, the following matrix inverse identity for
a $2\times 2$ block-partitioned positive definite matrix can be used:
\begin{equation} \label{eq:2x2blockinv}
 \begin{pmatrix}
  \bOmega_{11} & \bOmega_{12} \\
  \bOmega_{21} & \bOmega_{22} \\ \end{pmatrix}^{-1}
  = \begin{pmatrix}
   \bA^{-1} & -\bA^{-1}\bOmega_{12}\bOmega_{22}^{-1} \\
  -\bOmega_{22}^{-1} \bOmega_{21}\bA^{-1} &
     \bOmega_{22}^{-1} + \bOmega_{22}^{-1}\bOmega_{21} \bA^{-1}\bOmega_{12}\bOmega_{22}^{-1} \\
  \end{pmatrix},
\end{equation}
where $\bA=\bOmega_{11\cdot2}=\bOmega_{11} - \bOmega_{12}\bOmega_{22}^{-1}\bOmega_{21}$.

\begin{example}\label{ex:3_dimensional_VAR(1)}
(3 dimensional VAR(1) process). Let $\{\bZ_t\}_{t>0}$ follow a 3-dimensional VAR(1) model
with partition $\cbra{\{1,2\},\{3\}}$, i.e., $\bZ_{S_1,t} = \rbra{Z_{1,t}, Z_{2,t}}^T$ and $\bZ_{S_2,t} = (Z_{3,t})$.
The correlation matrix of random vector $\rbra{Z_{1,t}, Z_{2,t}, Z_{1,t-1}, Z_{2,t-1}, Z_{3,t}, Z_{3,t-1}}^T$ is
\begin{equation} \label{eq:trivariateVAR1cor}
R_{\cbra{\{1,2\}, \{3\}}} =
\begin{pmatrix}
\begin{BMAT}{c:c}{c:c}
  R_{\{1,2\}}         & R_{\{1,2\},\{3\}} \\
  R^T_{\{1,2\},\{3\}} & R_{\{3\}}
\end{BMAT}
\end{pmatrix}
=
\begin{pmatrix}
\begin{BMAT}{cc.cc:cc}{cc.cc:cc}
  1           & \rho_{12,0} & \rho_{11,1} & \rho_{12,1} & \rho_{13,0} & \rho_{13,1} \\
  \rho_{12,0} & 1           & \rho_{12,-1}& \rho_{22,1} & \rho_{23,0} & \rho_{23,1} \\
  \rho_{11,1} & \rho_{12,-1}& 1           & \rho_{12,0} & \rho_{13,-1}& \rho_{13,0} \\
  \rho_{12,1} & \rho_{22,1} & \rho_{12,0} & 1           & \rho_{23,-1}& \rho_{23,0} \\
  \rho_{13,0} & \rho_{23,0} & \rho_{13,-1}& \rho_{23,-1}& 1           & \rho_{33,1} \\
  \rho_{13,1} & \rho_{23,1} & \rho_{13,0} & \rho_{23,0} & \rho_{33,1} & 1
  \addpath{(5,2,.)uuuu}
  \addpath{(0,1,.)rrrr}
\end{BMAT}
\end{pmatrix}.
\end{equation}

Under the constraint that $\{\bZ_t\}_{t>0}$ is closed under margins
with respect to
$\cbra{\{1,2\},\{3\}}$, VAR(1) sub-process $\{\bZ_{S_1,t}\}_{t>0}$ can be specified by the following positive definite matrices:
\[
    R_{\{1,2\}} =
    \begin{pmatrix}
    \begin{BMAT}{c.c}{c.c}
      \bSigma_{11,0}   & \bSigma_{11,1} \\
      \bSigma_{11,1}^T & \bSigma_{11,0}
    \end{BMAT}
    \end{pmatrix}=
    \begin{pmatrix}
    \begin{BMAT}{cc.cc}{cc.cc}
      1           & \rho_{12,0} & \rho_{11,1} & \rho_{12,1}\\
      \rho_{12,0} & 1           & \rho_{12,-1}& \rho_{22,1}\\
      \rho_{11,1} & \rho_{12,-1}& 1           & \rho_{12,0}\\
      \rho_{12,1} & \rho_{22,1} & \rho_{12,0} & 1\\
    \end{BMAT}
    \end{pmatrix},
\]
and VAR(1) sub-process $\{\bZ_{S_2,t}\}_{t>0}$ can be specified by
\[
    R_{\{3\}} =
    \begin{pmatrix}
    \begin{BMAT}{c.c}{c.c}
      \bSigma_{22,0}   & \bSigma_{22,1} \\
      \bSigma_{22,1}^T & \bSigma_{22,0}
    \end{BMAT}
    \end{pmatrix}=
    \begin{pmatrix}
    \begin{BMAT}{c.c}{c.c}
      1 & \rho_{33,1} \\
      \rho_{33,1} & 1
    \end{BMAT}
    \end{pmatrix}.
\]

For other parameters, $\bSigma_{12,0} = \rbra{
\begin{smallmatrix}
  \rho_{13,0} \\
  \rho_{23,0}
\end{smallmatrix}}$ measures the contemporaneous dependence between $\{\bZ_{S_1,t}\}_{t>0}$ and $\{\bZ_{S_2,t}\}_{t>0}$,
$\bSigma_{12,-1}=\rbra{
\begin{smallmatrix}
  \rho_{13,-1} \\
  \rho_{23,-1}
\end{smallmatrix}}$ and
$\bSigma_{12,1}=\rbra{
\begin{smallmatrix}
  \rho_{13,1} \\
  \rho_{23,1}
\end{smallmatrix}}$ measures the cross-sectional dependence between $\{\bZ_{S_1,t}\}_{t>0}$ and $\{\bZ_{S_2,t}\}_{t>0}$ at lag $-1$ and $1$, respectively.
If matrix $R_{\cbra{\{1,2\}, \{3\}}}$ in \cref{eq:trivariateVAR1cor} is positive definite,
then from the conditional expectation formula for multivariate Gaussian distributions,
the coefficient matrix of 3-dimensional VAR(1) process
$\{(Z_{1,t},Z_{2,t},Z_{3,t})\}_{t>0}$ is
\begin{eqnarray}
  \bPhi &=&
 \begin{pmatrix}
  \bSigma_{11,1} & \bSigma_{12,1} \\
  \bSigma_{12,-1}^T & \rho_{33,1} \\
 \end{pmatrix}
 \begin{pmatrix}
  \bSigma_{11,0} & \bSigma_{12,0} \\
  \bSigma_{12,0}^T & 1 \\
 \end{pmatrix}^{-1} \nonumber \\
 &=& \begin{pmatrix}
  \bSigma_{11,1} & \bSigma_{12,1} \\
  \bSigma_{12,-1}^T & \rho_{33,1} \\
 \end{pmatrix}
 \begin{pmatrix}
   \bA^{-1} & -\bA^{-1}\bSigma_{12,0} \\
  -\bSigma_{12,0}^T\bA^{-1} &
     1 + \bSigma_{12,0}^T \bA^{-1}\bSigma_{12,0} \\
 \end{pmatrix},
  \label{eq:trivariateVAR1Phi}
\end{eqnarray}
where $\bA=\bSigma_{11,0}-\bSigma_{12,0} \bSigma_{12,0}^T$.
In the above, \cref{eq:2x2blockinv} is applied.

With $R_{\cbra{1,2}}$ and $R_{\cbra{3}}$ fixed, conditions on
$\bSigma_{12,-1}, \bSigma_{12,0}, \bSigma_{12,1}$ can lead to closure under margins
with respect to the specified partition.
The rule of conditional expectation for multivariate Gaussian distributions gives
\[
   \bPhi_{1,1} = \bPsi_{1,1} = \bSigma_{11,1}\bSigma_{11,0}^{-1}\ \mbox{and}\
   \bPhi_{2,1} = \bPsi_{2,1} = \bSigma_{22,1}\bSigma_{22,0}^{-1} = \rho_{33,1}.
\]
Note that $\bPhi_{1,1}$ here is the same as in \cref{eq:bivariateVAR1Phi}.
The conditions in
\cref{eq:G1} to \cref{eq:H2}
become:
\begin{equation}
 \begin{aligned}
 G_1D_1=0 \Rightarrow &\ \bPhi_{1,1}\bSigma_{12,0}-\bSigma_{12,1}=\bm{0},\\
 H_1D_1=0 \Rightarrow &\ -\bSigma_{12,-1}+\bPsi_{1,1}\bSigma_{12,0}=\bm{0}.\\
 (G_2L_2)(L_2D_2)=0 \Rightarrow &\ -\bSigma_{12,-1}^T+\bPhi_{2,1}\bSigma_{12,0}^T=\bm{0}\Rightarrow
-\bSigma_{12,-1}+\bSigma_{12,0}\rho_{33,1}=\bm{0},\\
 (H_2L_2)(L_2D_2)=0 \Rightarrow &\ \bPsi_{21}\bSigma_{12,0}^T-\bSigma_{12,1}^T=\bm{0} \Rightarrow
\bSigma_{12,0}\rho_{33,1}-\bSigma_{12,1}=\bm{0}.
 \end{aligned}
\end{equation}

In Case 1, based on $G_1$ and $G_2$,
$\bSigma_{12,1}=\bPhi_{1,1}\bSigma_{12,0}=\bSigma_{11,1}\bSigma_{11,0}^{-1}
\bSigma_{12,0}$ and
$\bSigma_{12,-1}=\bSigma_{12,0}\rho_{33,1}$.
If \cref{eq:trivariateVAR1cor} is positive definite,  then
\cref{eq:trivariateVAR1Phi} becomes
\begin{eqnarray*}
 \bPhi&=& \begin{pmatrix}
  \bSigma_{11,1} & \bSigma_{11,1}\bSigma_{11,0}^{-1}\bSigma_{12,0} \\
   \rho_{33,1}\bSigma_{12,0}^T & \rho_{33,1} \\
 \end{pmatrix}
 \begin{pmatrix}
   \bA^{-1} & -\bA^{-1}\bSigma_{12,0} \\
  -\bSigma_{12,0}^T\bA^{-1} &
     1 + \bSigma_{12,0}^T \bA^{-1}\bSigma_{12,0} \\
 \end{pmatrix} \\
 &=& \begin{pmatrix}
  \bSigma_{11,1}\bSigma_{11,0}^{-1} & \bm{0} \\
  \bm{0}  & \rho_{33,1} \\
 \end{pmatrix},
\end{eqnarray*}
after some algebraic simplifications.
This leads to a diagonal coefficient matrix so that the two sub-processes
are marginal VAR(1) and AR(1) respectively with dependence coming from the
innovation vector.

In Case 3, based on $G_1$ and $H_2$,
the general solution has $\bSigma_{12,0}=\bSigma_{12,1}=\bm{0}$ with  $\bSigma_{12,-1}$ fixed.
If matrix $R_{\cbra{\{1,2\}, \{3\}}}$ in \cref{eq:trivariateVAR1cor}
is positive definite, then the coefficient matrix in
\cref{eq:trivariateVAR1Phi} becomes
  $$\bPhi =
 \begin{pmatrix}
  \bSigma_{11,1} & \bm{0} \\
  \bSigma_{12,-1}^T & \rho_{33;1} \\
 \end{pmatrix}
 \begin{pmatrix}
  \bSigma_{11,0} & \bm{0} \\
  \bm_{0} & 1 \\
 \end{pmatrix}^{-1}
 =\begin{pmatrix}
  \bSigma_{11,1} \bSigma_{11,0}^{-1} & \bm{0} \\
  \bSigma_{12,-1}^T \bSigma_{11,0} & \rho_{33;1} \\
 \end{pmatrix},
 $$
and clearly $\{Z_{1,t},Z_{2,t}\}_{t>0}$ is VAR(1).

In Case 4, based on $H_1$ and $G_2$,
the general solution has $\bSigma_{12,0}=\bSigma_{12,-1}=\bm{0}$ with  $\bSigma_{12,1}$ fixed.
If matrix $R_{\cbra{\{1,2\}, \{3\}}}$ in \cref{eq:trivariateVAR1cor}
is positive definite, then
  $$\bPhi =
 \begin{pmatrix}
  \bSigma_{11,1} & \bSigma_{12,1} \\
  \bm{0} & \rho_{33;1} \\
 \end{pmatrix}
 \begin{pmatrix}
  \bSigma_{11,0} & \bm{0} \\
  \bm_{0} & 1 \\
 \end{pmatrix}^{-1}
 =\begin{pmatrix}
  \bSigma_{11,1} \bSigma_{11,0}^{-1} & \bSigma_{12,1} \\
  \bm_{0} & \rho_{33;1} \\
 \end{pmatrix},
 $$
and clearly $\{Z_{3,t}\}_{t>0}$ is AR(1).

The interesting case is Case 2 based on $H_1$ and $H_2$,
with $\bSigma_{12,-1}=\Phi_{1,1}\bSigma_{12,0}$ and
$\bSigma_{12,1}=\bSigma_{12,0}\rho_{33,1}$,
as it leads to a non-diagonal $\bPhi$ coefficient matrix in general.
As an illustration, setting
$\bSigma_{11,0}=\rbra{
\begin{smallmatrix}
  1 & 0.9 \\
  0.9 & 1
\end{smallmatrix}}$,
$\bSigma_{22,0}=(0.9)$,
$\bSigma_{11,1}=\rbra{
\begin{smallmatrix}
  0.8 & 0.7 \\
  0.9 & 0.8
\end{smallmatrix}}$, and
$\bSigma_{12,0}=\rbra{
\begin{smallmatrix}
  0.5 \\
  0.5
\end{smallmatrix}}$.
makes $R_{\cbra{\{1,2\}, \{3\}}}$ positive definite and $\bPhi$ non-diagonal.
\end{example}

The two examples above also indicate the fact that some non-diagonal blocks of the coefficient matrix would be $\bm{0}$ if at least a sub-process satisfies Condition 1. Indeed, a general conclusion about the model can be drawn: if a sub-process satisfies Condition 1, then, in the VAR representation of the original time series, the regression coefficients of the sub-process on the other sub-processes will be 0. Therefore, Case 4 when both sub-processes satisfy Condition 1 corresponds to the special situation of diagonal coefficient matrices in blocks. To show this, the coefficient matrix $\bPhi_i$ for $1\leq i \leq k$ is blocked as
\[
    \bPhi_i =
    \begin{pmatrix}
      \bPhi_{i, S_1} & \bPhi_{i, S_1, S_2} \\
      \bPhi_{i, S_2, S_1} & \bPhi_{i, S_1}
    \end{pmatrix},
\]
where the dimension of $\bPhi_{i, S_1}$ is $d_1\times d_1$. Then $\bZ_{S_1,t}$ can be represented as
\[
    \bZ_{S_1,t} = \sum_{i=1}^{k}\bPhi_{i, S_1}\bZ_{S_1,t-i} + \sum_{i=1}^{k}\bPhi_{i, S_1,S_2}\bZ_{S_2,t-i} + \bepsilon_{S_1,t},
\]
where $\bepsilon_{S_1,t}$ is the sub-vector consisting of the first $d_1$ elements of $\bepsilon_{t}$. Furthermore, $\bPhi_{1, S_1,S_2}=\cdots=\bPhi_{k, S_1,S_2}=\bm{0}$ is equivalent to
\[\sbra{\bZ_{S_1,t} \bot \rbra{\bZ_{S_2,t-1}^T,\dots,\bZ_{S_2,t-k}^T}^T} \big| \bZ_{S_1,t-1},\dots,\bZ_{S_1,t-k},\]
which is exactly Condition 1 for $\cbra{\bZ_{S_1,t}}_{t>0}$. Similarly, Condition 1 for $\cbra{\bZ_{S_2,t}}_{t>0}$ is equivalent to $\bPhi_{1, S_2,S_1}=\cdots=\bPhi_{k, S_2,S_1}=\bm{0}$. And $\bPhi_i$ would be diagonal in blocks in the case that both sub-processes satisfiy Condition 1 since $\bPhi_{i, S_1,S_2}=\bPhi_{i, S_2,S_1}=\bm{0}$ for $1\leq i \leq k$.

\subsubsection{Partitions with multiple sub-processes}
\label{sec:partition_with_multi}

For partitions with multiple sub-processes,  let $\cbra{S_1, \dots, S_n}$ be a partition of $\cbra{1,\dots,d}$ and let $d_i\ge1$ be the cardinality of $S_i$.
Reorder the row and columns of $R$ to get the correlation matrix of the random vector $\rbra{\bZ_{S_1,t}^T,\dots, \bZ_{S_1,t-k}^T, \dots, \bZ_{S_n,t}^T,\dots, \bZ_{S_n,t-k}^T}^T$:
\begin{equation}\label{eq:correlation_matrix_of_S1_to_Sn}
    R_{\cbra{S_1, \dots, S_n}} =
    \begin{pmatrix}
        R_{S_1}       & R_{S_1,S_2}       & \cdots & R_{S_1, S_n} \\
        R_{S_1,S_2}^T & R_{S_2}           & \cdots & R_{S_2,S_{n-1}} \\
        \vdots        & \vdots            & \ddots & \vdots \\
        R_{S_1,S_n}^T & R_{S_2,S_{n-1}}^T & \cdots & R_{S_n}
    \end{pmatrix},
\end{equation}
where
\begin{equation}\label{eq:Sigma_to_Ri}
    R_{S_i} =
    \begin{pmatrix}
        \bSigma_{ii ,0}  & \bSigma_{ii,1}     & \cdots & \bSigma_{ii,k} \\
        \bSigma_{ii,1}^T & \bSigma_{ii,0}     & \cdots & \bSigma_{ii,k-1} \\
        \vdots           & \vdots             & \ddots & \vdots \\
        \bSigma_{ii,k}^T & \bSigma_{ii,k-1}^T & \cdots & \bSigma_{ii,0}
    \end{pmatrix},
    R_{S_i,S_j} =
    \begin{pmatrix}
        \bSigma_{ij,0}   & \bSigma_{ij,1}   & \cdots & \bSigma_{ij,k-1} \\
        \bSigma_{ij,-1}  & \bSigma_{ij,0}   & \cdots & \bSigma_{ij,k-2} \\
        \vdots           & \vdots           & \ddots & \vdots \\
        \bSigma_{ij,1-k} & \bSigma_{ij,2-k} & \cdots & \bSigma_{ij,0}
    \end{pmatrix}.
\end{equation}
Here $\bSigma_{ij, l}$ for $1\leq i,j \leq n$, which denote the covariance matrices between $\bZ_{S_i,t}$ and $\bZ_{S_j, t-l}$, are treated as the parameter entries of $R_{\cbra{S_1, \dots, S_n}}$.
Notice that $R_{S_i}$ is the correlation matrix of $\cbra{\bZ_{S_i,t}^T, \dots, \bZ_{S_i,t-k}^T}^T$ and $ R_{S_i,S_j}$ is the correlation matrix between $\cbra{\bZ_{S_i,t}^T, \dots, \bZ_{S_i,t-k}^T}^T$ and $\cbra{\bZ_{S_j,t}^T, \dots, \bZ_{S_j,t-k}^T}^T$.
It means that when $\cbra{\bZ_t}_{t>0}$ is closed under margins with respect to partition $\cbra{S_1, \dots, S_n}$, $R_{S_i}$ actually specifies the VAR($k$) sub-process $\cbra{\bZ_{S_i,t}}_{t>0}$ with contemporaneous correlation $\bSigma_{ii ,0}$ and serial correlation $\bSigma_{ii,1}, \cdots, \bSigma_{ii,k}$.
And $R_{S_i,S_j}$ models the dependence between $\cbra{\bZ_{S_i,t}}_{t>0}$ and $\cbra{\bZ_{S_j,t}}_{t>0}$ with contemporaneous correlation $\bSigma_{ij,0}$ and cross-sectional correlations $\bSigma_{ij,-k}, \cdots, \bSigma_{ij,-1}, \bSigma_{ij,1}, \cdots, \bSigma_{ij,k}$.
Therefore, the parameter entries of $R_{\cbra{S_1, \dots, S_n}}$ can be further divided into two parts given the partition $\cbra{S_1,\dots,S_n}$: the entries of $R_{S_1}, \dots,  R_{S_n}$ that model all individual VAR($k$) sub-processes, and the entries of $R_{S_i, S_j}$ for $i \neq j$ that model the dependence structure between the sub-processes.
To make $\cbra{\bZ_{t}}_{t>0}$ closed under margins with respect to $\{S_1,\dots,S_n\}$, we still hold the parameters of $R_{S_1}, \dots,  R_{S_n}$ fixed while investigating the constraints on the parameters of $R_{S_i, S_j}$ for $i \neq j$.
Note that we only need to consider the case of $i<j$ since $\bSigma_{ji,k}=\bSigma_{ij,-k}^T$.

We then deal with the case of multiple sub-processes by considering all pairs of distinct sub-processes with the approach in \cref{sec:partition_with_two} above.
Indeed, the two conditions for the process $\cbra{\bZ_{S_i,t}}_{t>0}$ in \cref{thm:sufficient_condition} can be written as
\begin{enumerate}
  \item $\sbra{\bZ_{S_i,t} \bot \rbra{\bZ_{S_j,t-1}^T,\dots,\bZ_{S_j,t-k}^T}^T} \big| \bZ_{S_i,t-1},\dots,\bZ_{S_i,t-k}\ \mbox{for}\ j\neq i$,
  \item $\sbra{\bZ_{S_i,t-k-1} \bot \rbra{\bZ_{S_j,t-1}^T,\dots,\bZ_{S_j,t-k}^T}^T} \big| \bZ_{S_i,t-1},\dots,\bZ_{S_i,t-k}\ \mbox{for}\ j\neq i$.
\end{enumerate}
It indicates that Condition 1 or 2 holds for sub-process $\cbra{\bZ_{S_i,t}}_{t>0}$ if and only if the same condition as in \cref{eq:thm3_for_S1S2} holds for any pairs of sub-processes $\cbra{\bZ_{S_i,t}}_{t>0}$ and $\cbra{\bZ_{S_j,t}}_{t>0}$ for $1\leq j \leq n$. For simplicity, let $c_i\in \cbra{1,2}$ label the index of the above condition that $\cbra{\bZ_{S_i,t}}_{t>0}$ satisfies.
Then the sufficient condition for pair $(i,j)\in \cbra{1,\dots,n}^2$ with $i\neq j$ is similar to the results of partitions with two sub-processes in \cref{sec:partition_with_two}.
In particular, the sufficient condition for $(i,j)$ can be expressed as the constraints on the dependence parameters between two sub-processes $\cbra{\bZ_{S_i,t}}_{t>0}$ and $\cbra{\bZ_{S_j,t}}_{t>0}$.
For $1\leq p \leq n$, the coefficient matrices $G_{p}$ and $H_{p}$ are defined as
\begin{equation}\label{eq:definition_of_G}
G_{p} =
\begin{pmatrix}
  0      & \bPhi_{p,k} & \dots       & \bPhi_{p,2} & \bPhi_{p,1} & -I_{d_p}  & 0           & \cdots      & 0\\
  0      & 0           & \bPhi_{p,k} & \dots       & \bPhi_{p,2} & \bPhi_{p,1} & -I_{d_p}  & \cdots      & 0 \\
  \vdots & \ddots      & \ddots      & \ddots      & \ddots      & \ddots      &   \ddots    & \ddots      & \vdots \\
  0      & \cdots      & 0           & 0           & \bPhi_{p,k} & \cdots      & \bPhi_{p,2} & \bPhi_{p,1} & -I_{d_p}
 \end{pmatrix},
\end{equation}
where
\begin{equation}
\begin{pmatrix}
   \bPhi_{p,1}^T \\
   \bPhi_{p,2}^T \\
   \vdots \\
   \bPhi_{p,k}^T
 \end{pmatrix}^T =
 \begin{pmatrix}
  \bSigma_{pp,1}^T \\
  \bSigma_{pp,2}^T \\
  \vdots \\
  \bSigma_{pp,k}^T
 \end{pmatrix}^T
 \begin{pmatrix}
  \bSigma_{pp,0}     & \bSigma_{pp,1}     & \cdots & \bSigma_{pp,k-1} \\
  \bSigma_{pp,1}^T   & \bSigma_{pp,0}     & \cdots & \bSigma_{pp,k-2} \\
  \vdots             & \vdots             & \ddots & \vdots \\
  \bSigma_{pp,k-1}^T & \bSigma_{pp,k-2}^T & \cdots & \bSigma_{pp,0}
 \end{pmatrix}^{-1}
\end{equation}
and
\begin{equation}\label{eq:definition_of_H}
H_{p} =
 \begin{pmatrix}
  -I_{d_p} & \bPsi_{p,k} & \dots & \bPsi_{p,2} & \bPsi_{p,1} & 0 & 0 & \cdots & 0\\
  0 & -I_{d_p} & \bPsi_{p,k} & \dots & \bPsi_{p,2} & \bPsi_{p,1} & 0 & \cdots & 0 \\
  \vdots  & \ddots & \ddots & \ddots & \ddots & \ddots &   \ddots & \ddots & \vdots \\
  0   & \cdots & 0 & -I_{d_p} & \bPsi_{p,k} & \cdots & \bPsi_{p,2} & \bPsi_{p,1} & 0
 \end{pmatrix},
\end{equation}
where
\begin{equation}
\begin{pmatrix}
   \bPsi_{p,1}^T \\
   \bPsi_{p,2}^T \\
   \vdots \\
   \bPsi_{p,k}^T
 \end{pmatrix}^T =
 \begin{pmatrix}
  \bSigma_{pp,-k}^T \\
  \bSigma_{pp,1-k}^T \\
  \vdots \\
  \bSigma_{pp,-1}^T
 \end{pmatrix}^T
 \begin{pmatrix}
  \bSigma_{pp,0}     & \bSigma_{pp,1}     & \cdots & \bSigma_{pp,k-1} \\
  \bSigma_{pp,1}^T   & \bSigma_{pp,0}     & \cdots & \bSigma_{pp,k-2} \\
  \vdots             & \vdots             & \ddots & \vdots \\
  \bSigma_{pp,k-1}^T & \bSigma_{pp,k-2}^T & \cdots & \bSigma_{pp,0}
 \end{pmatrix}^{-1}.
\end{equation}
Since $R_{S_p}$ is the correlation matrix of $\cbra{\bZ_{S_p,t}^T, \dots, \bZ_{S_p,t-k}^T}^T$, it can also be checked that $\bPhi_{p,1},\dots,\bPhi_{p,k}$ are the coefficient matrices of VAR($k$) process $\cbra{\bZ_{S_p,t}}_{t>0}$ when $\cbra{\bZ_{t}}_{t>0}$ is closed under margins with respect to partition $\{S_1,\dots,S_n\}$.
It follows that for pair $(i,j)$, all combinations of Condition 1 and 2 for the two sub-processes of $S_i$ and $S_j$ can be written as
\begin{equation}\label{eq:conditions_for_Si_S_j}
    \begin{aligned}
    (c_i,c_j) = (1,1):\ &  G_{i}D_{ij} = \bm{0}\ \mbox{and}\ (G_{j}L_j)(L_jD_{ji}) = \bm{0},\\
    (c_i,c_j) = (1,2):\ &  G_{i}D_{ij} = \bm{0}\ \mbox{and}\ (H_{j}L_j)(L_jD_{ji}) = \bm{0},\\
    (c_i,c_j) = (2,1):\ &  H_{i}D_{ij} = \bm{0}\ \mbox{and}\ (G_{j}L_j)(L_jD_{ji}) = \bm{0},\\
    (c_i,c_j) = (2,2):\ &  H_{i}D_{ij} = \bm{0}\ \mbox{and}\ (H_{j}L_j)(L_jD_{ji}) = \bm{0}.\\
    \end{aligned}
\end{equation}
where $D_{i,j}=\rbra{ \bSigma_{ij, -k}^T, \dots, \bSigma_{ij,-1}^T, \bSigma_{ij,0}^T, \bSigma_{ij,1}^T, \dots, \bSigma_{ij,k}^T }^T$,
\begin{multline}
    D_{j,i} = \rbra{ \bSigma_{ji, -k}^T, \dots, \bSigma_{ji,-1}^T, \bSigma_{ji,0}^T, \bSigma_{ji,1}^T, \dots, \bSigma_{ji,k}^T }^T \\
    = \rbra{ \bSigma_{ij, k}, \dots, \bSigma_{ij,1}, \bSigma_{ij,0}, \bSigma_{ij,-1}, \dots, \bSigma_{ij,-k} }^T,
\end{multline}
and $L_j = J_{2k+1} \otimes I_{|S_j|}$ where $J_{2k}$ is the
$(2k+1)$-dimensional exchange matrix whose elements in the anti-diagonal are 1 and all other elements are zero.

It follows that given the condition labels $c_1,\dots,c_n$ and
$R_{S_1},\dots,R_{S_n}$ that model all individual VAR($k$) sub-processes,
$\bSigma_{ij,-k}, \dots, \bSigma_{ij,k}$ for all pairs $(i,j)$ can be
parameterized by the given or fixed parameters among them according to
Cases 1--4 in \cref{appd:multi_sub_processes}.
Therefore the correlation matrix $R_{\{S_1,\dots,S_n\}}$ can be characterized by three groups of parameters: the condition labels $c_1,\dots,c_n$, the entries of $R_{S_1}, \dots, R_{S_n}$ that model all individual VAR($k$) sub-processes, and the corresponding fixed parameters.
Note that the fixed parameter is $\bSigma_{ij, 0}$ if $c_i=c_j=1$ or $c_i=c_j=2$, $\bSigma_{ij, -k}$ if $c_i=1,c_j=2$, and $\bSigma_{ij, k}$ if $c_i=2,c_j=1$ for each pair $(i,j)$ of $1\leq i < j\leq n$.
Moreover, an extra necessary constraint that $R_{\{S_1,\dots,S_n\}}$ is positive definite should always be guaranteed.
Then correlation matrix $R$ in \cref{eq:multivariate_Gaussian_copula} can be obtained by reordering the rows and columns of $R_{\{S_1,\dots,S_n\}}$.

Note that for the case of multiple sub-processes, we can still draw the general conclusion that if sub-process $\cbra{\bZ_{S_i,t}}_{t>0}$ satisfies Condition 1, then in the VAR representation of the original time series, the regression coefficients of $\bZ_{S_i,t}$ on $\bZ_{S_j,t-1},\dots,\bZ_{S_j,t-k}$ for $i\neq j$ will be $\bm{0}$.
It means that all coefficient matrices of the original VAR process will be diagonal in blocks if all sub-processes fulfill Condition 1, and all sub-processes should satisfy Condition 2 if all non-diagonal blocks of the coefficient matrices of the original VAR process are required to be non-zero.

The next example shows how to deal with a partition with three sub-processes based on the results of partitions with two sub-processes, and general formulas for solving the linear equations of Cases 1 and 2 are derived in \cref{appd:multi_sub_processes}.

\begin{example}\label{ex:3_dimensional_VAR(1)_2}
(3-dimensional VAR(1) process). Let $\{\bZ_t\}_{t>0}$ follow a 3-dimensional VAR(2) model and the considered partition is $\cbra{\{1\}, \{2\},\{3\}}$, i.e., $\bZ_{S_1,t} = Z_{1,t}, \bZ_{S_2,t}=Z_{2,t}$ and $\bZ_{S_3,t} = Z_{3,t}$. The correlation matrix of $\rbra{Z_{1,t}, Z_{1,t-1}, Z_{2,t}, Z_{2,t-1}, Z_{3,t}, Z_{3,t-1}}^T$ is
\[
\begin{pmatrix}
\begin{BMAT}{c:c:c}{c:c:c}
  R_{\{1\}}         & R_{\{1\},\{2\}}   & R_{\{1\},\{3\}}\\
  R^T_{\{1\},\{2\}} & R_{\{2\}}         & R_{\{2\},\{3\}}\\
  R^T_{\{1\},\{3\}} & R^T_{\{2\},\{3\}} & R_{\{3\}}
\end{BMAT}
\end{pmatrix}
=
\begin{pmatrix}
\begin{BMAT}{cc:cc:cc}{cc:cc:cc}
  1           & \rho_{11,1} & \rho_{12,0} & \rho_{12,1} & \rho_{13,0} & \rho_{13,1} \\
  \rho_{11,1} & 1           & \rho_{12,-1}& \rho_{12,0} & \rho_{13,-1}& \rho_{13,0} \\
  \rho_{12,0} & \rho_{12,-1}&           1 & \rho_{22,1} & \rho_{23,0} & \rho_{23,1} \\
  \rho_{12,1} & \rho_{12,0} & \rho_{22,1} & 1           & \rho_{23,-1}& \rho_{23,0} \\
  \rho_{13,0} & \rho_{13,-1}& \rho_{23,0} & \rho_{23,-1}& 1           & \rho_{33,1} \\
  \rho_{13,1} & \rho_{13,0} & \rho_{23,1} & \rho_{23,0} & \rho_{33,1} & 1
\end{BMAT}
\end{pmatrix}.
\]
It can be verified that given $\{\bZ_t\}_{t>0}$ is closed under margins
with respect to
$\cbra{\{1\}, \{2\},\{3\}}$, VAR(1) sub-processes $\{\bZ_{S_1,t}\}_{t>0}$, $\{\bZ_{S_2,t}\}_{t>0}$, $\{\bZ_{S_3,t}\}_{t>0}$ can be specified by $R_{\{1\}}$, $R_{\{2\}}$, and $R_{\{3\}}$ respectively.
Then for each pair $(i,j)$ for $1\leq i < j \leq 3$, we fix $R_{\{i\}},R_{\{j\}}$ and impose the constraint on $R_{\{i\},\{j\}}$.
Indeed, for two sub-processes $\{\bZ_{S_i,t}\}_{t>0}$ and $\{\bZ_{S_j,t}\}_{t>0}$ of pair $(i,j)$, according to \cref{ex:2_dimensional_VAR(1)}, the conditions can be written as

\medskip

1: $(c_i,c_j) = (1,1)$: set $\rho_{ij,0}$ as the parameter and
$\rho_{ij,1}=\rho_{ii,1}\rho_{ij,0}$, $\rho_{ij,-1}=\rho_{jj,1}\rho_{ij,0}$;\\
\indent 2: $(c_i,c_j) = (1,2)$: set $\rho_{ij,-1}$ as the parameter and let $\rho_{ij,1}=\rho_{ij,0}=0$;\\
\indent 3: $(c_i,c_j) = (2,1)$: set $\rho_{ij,1}$ as the parameter and let $\rho_{ij,-1}=\rho_{ij,0}=0$;\\
\indent 4: $(c_i,c_j) = (2,2)$: set $\rho_{ij,0}$ as the parameter and
$\rho_{ij,-1}=\rho_{ii,1}\rho_{ij,0}$, $\rho_{ij,1}=\rho_{jj,1}\rho_{ij,0}$.

\medskip

Combining all results above, there are conditions on
$R_{\{1\},\{2\}}=\rbra{
\begin{smallmatrix}
  \rho_{12,0} & \rho_{12,1} \\
  \rho_{12,-1} & \rho_{12,0}
\end{smallmatrix}
}$,
$R_{\{1\},\{3\}}=\rbra{
\begin{smallmatrix}
  \rho_{13,0} & \rho_{13,1} \\
  \rho_{13,-1} & \rho_{13,0}
\end{smallmatrix}
}$, and
$R_{\{2\},\{3\}}=\rbra{
\begin{smallmatrix}
  \rho_{23,0} & \rho_{23,1} \\
  \rho_{23,-1} & \rho_{23,0}
\end{smallmatrix}
}$ for any given condition label. For instance, if the condition label
$(c_1,c_2,c_3)$ is $(1,2,2)$, then the above results for three pairs
$(c_1,c_2)=(1,2)$, $(c_1,c_3)=(1,2)$, and $(c_2,c_3)=(2,2)$ should be
picked, i.e., $\rho_{12,-1}, \rho_{13,-1}, \rho_{23,0}$ should be set as the
fixed parameters and other dependence
 parameters between the sub-processes can be solved using $\rho_{12,1}=\rho_{12,0}=\rho_{13,1}=\rho_{13,0}=0$, $\rho_{23,-1}=\rho_{22,1}\rho_{23,0}$, and $\rho_{23,1}=\rho_{33,1}\rho_{23,0}$.
But when $(c_1,c_2,c_3)$ is $(2,2,2)$ so that all coefficient matrices are non-diagonal in general,
 the dependence parameters should satisfy
\[
    \rho_{12,-1}=\rho_{11,1}\rho_{12,0}, \quad \rho_{12,1}=\rho_{22,1}\rho_{12,0},
    \quad \rho_{13,-1}=\rho_{11,1}\rho_{13,0}, \\
    \quad \rho_{13,1}=\rho_{33,1}\rho_{13,0},\quad
    \rho_{23,-1}=\rho_{22,1}\rho_{23,0}, \quad \rho_{23,1}=\rho_{33,1}\rho_{23,0}.
\]
A compatible numerical example in this case is $\rho_{11,1}=0.6$, $\rho_{22,1}=0.7$, $\rho_{33,1}=0.8$, and $\rho_{12,0}=\rho_{13,0}=\rho_{23,0}=0.5$.
\end{example}

An interesting case is when the number of sub-processes in the partition is exactly $d$, and so all the sub-processes are univariate. Then, margins of any dimension and any univariate components of the $d$-dimensional VAR($k$) model are also VAR($k$) processes. Moreover, according to our discussion about the blocks of coefficient matrices at the end of \cref{sec:partition_with_two}, the coefficient matrices are all diagonal in this case if all univariate sub-processes satisfy Condition 1. There will be at least one coefficient matrix with non-zero non-diagonal entries if two conditions are met: at least one univariate sub-process satisfies Condition 2, and at least one dependence parameter between a sub-process satisfying Condition 2 and the other sub-processes is non-zero.

\section{Parameter estimation}
\label{sec:estimation}

In this section, details of the maximum likelihood estimation of the
margin-closed stationary multivariate time series model are given.
If the stationary joint distributions of VAR($k$) Gaussian model,
with possible margin closure under sub-processes, is
used as a multivariate copula, then univariate margins can be first fitted
with parametric families, before probability integral transforms to
standard Gaussian margins to estimate the dependence parameters of the VAR($k$) model.
Note that all latent VAR($k$) models are parameterized by the block Toeplitz correlation matrices of $k + 1$ consecutive observations in our fitting procedure below.
To get the VAR representations of the models, the Durbin-Levinson algorithm can be applied,
see Section 11.4 in \cite{BrockwellAndDavis}.

The joint probability density function (PDF) of consecutive $k+1$ observations in the time series is
\begin{equation}\label{eq:pdf_of_model}
\begin{aligned}
  &f_{\bX_{t:(t-k)}}\rbra{\bx_t,\dots,\bx_{t-k}; \boldeta_1, \dots,\boldeta_d, R} \\
  = &c_{\bX_{t:(t-k)}}\rbra{u_{1,t}, \dots, u_{d,t}, \dots, u_{1,t-k}, \dots, u_{d,t-k};R}\prod_{l=1}^{k}\prod_{i=1}^{d} f_i\rbra{x_{i,t-l};\boldeta_i},
\end{aligned}
\end{equation}
where $\bx_{t-l} = \rbra{x_{1,t-l},\dots,x_{d,t-l}}^T$ is the realization of
$\bX_{t-l}$ for $0\leq l \leq k$, $f_i\rbra{\cdot;\boldeta_i}$ and
$F_i\rbra{\cdot;\boldeta_i}$ are the PDF and CDF of the univariate marginal component
$X_{i,t}$ for $1\leq i \leq d$, and $u_{i,t-l} = F_i\rbra{x_{i,t-l};\boldeta_i}$.
Then, the log-likelihood of the given realization of the time series $\rbra{\bx_1,\dots,\bx_T}$ is
\begin{equation}\label{eq:log_likelihood_of_data}
\begin{aligned}
  & \ell\rbra{\boldeta_1, \dots,\boldeta_d,R|\bx_1,\dots,\bx_T} \\
= & \sum_{t=k+1}^{T} \log f_{\bX_{t}|\bX_{(t-1):(t-k)}}\rbra{\bx_t|\bx_{t-1},\dots,\bx_{t-k}; \boldeta_1, \dots,\boldeta_d, R} \\
  & \hspace{3.5cm} + \sum_{t=1}^{k} \log f_{\bX_{t}|\bX_{(t-1):1}}\rbra{\bx_t|\bx_{t-1},\dots,\bx_{1}; \boldeta_1, \dots,\boldeta_d, R}
\end{aligned}
\end{equation}
where $f_{\bX_{t}|\bX_{(t-1):(t-k)}}$ is the conditional density of
$\bX_{t}$ given $\bX_{t-1},\dots,\bX_{t-k}$; this can be analytically
derived based on \cref{eq:pdf_of_model} and the conditional distributions of
multivariate Gaussian random vectors.

Note that in \cref{eq:log_likelihood_of_data}, $\boldeta_i, \dots,\boldeta_d $ are parameters of the univariate margins, and the correlation matrix $R$ should be parameterized by following the method in \cref{sec:partition_with_multi}.
The partition and condition labels are treated as hyperparameters to fit the model.
Then entries of $R$ can be divided into two groups: the parameters in $R_{S_1}, \dots,  R_{S_n}$ that model all individual VAR($k$) sub-processes, and the parameters in $R_{S_i, S_j}$ for $i<j$ that model the dependence structure between the sub-processes.
More precisely, the log-likelihood of a sub-process $\cbra{\bX_{S_i,t}}_{t>0}$ for $S_i=\cbra{s_{i,1},\dots,s_{i,m_i}}$ and realization $\bx_{S_i,t} = \rbra{x_{s_{i,1},t}, \dots, x_{s_{i,m_i},t}}^T$ in $1\leq t\leq T$ can actually be specified by $\boldeta_1, \dots,\boldeta_d$ and $R_{S_i}$:
\begin{equation}\label{eq:log_likelihood_of_subdata}
\begin{aligned}
  & \ell_{S_i}\rbra{\boldeta_1, \dots,\boldeta_d,R_{S_i}|\bx_1,\dots,\bx_T} \\
= & \sum_{t=k+1}^{T} \log f_{\bX_{S_i,t}|\bX_{S_i,(t-1):(t-k)}}\rbra{\bx_{S_i,t}|\bx_{S_i,t-1},\dots,\bx_{S_i,t-k}; \boldeta_1, \dots,\boldeta_d, R_{S_i}} \\
  & \hspace{1cm} + \sum_{t=1}^{k} \log f_{\bX_{S_i,t}|\bX_{S_i,(t-1):1}}\rbra{\bx_{S_i,t}|\bx_{S_i,t-1},\dots,\bx_{S_i,1}; \boldeta_1, \dots,\boldeta_d, R_{S_i}},
\end{aligned}
\end{equation}
where the conditional density $f_{\bX_{S_i,t}|\bX_{S_i,(t-1):(t-k)}}$ can be derived from the joint density
\begin{equation}\label{eq:pdf_of_submodel}
\begin{aligned}
  &f_{\bX_{S_i,t:(t-k)}}\rbra{\bx_{S_i,t},\dots,\bx_{S_i,t-k}; \boldeta_1, \dots,\boldeta_d, R_{S_i}} \\
  = &c_{\bX_{S_i,t:(t-k)}}\rbra{u_{s_{i,1},t}, \dots, u_{s_{i,m_i},t}, \dots, u_{s_{i,1},t-k}, \dots, u_{s_{i,m_i},t-k};R_{S_i}}\prod_{l=1}^{k}\prod_{i\in S_i} f_i\rbra{x_{i,t-l};\boldeta_i}
\end{aligned}
\end{equation}
by using the properties of the conditional distribution for multivariate Gaussian random vectors. Based on the division of the parameter set as well as the consistency and asymptotic normality of the
quasi maximum likelihood estimators based on log-likelihood of
marginal densities
proved by \cite{Francq2013MarginalLaw}, a multiple-stage procedure of estimation can be applied.

\medskip

\noindent \textbf{Step 1}. Estimate the univariate margin parameters $\boldeta_1, \dots,\boldeta_d$ by maximizing the quasi-likelihood of margins, i.e., $\hat{\boldeta}_i = \arg\max_{\boldeta_i} \sum_{t=1}^{T}\log f_i\rbra{x_{i,t};\boldeta_i}$ for $i=1,\dots,d$.

\noindent \textbf{Step 2}. For each sub-process $i\in\cbra{1,\dots,n}$ in the partition, hold $\boldeta_1, \dots,\boldeta_d$ fixed of their estimates obtained in step 1, and estimate $R_{S_i}$
individually through maximizing the objective function in \cref{eq:log_likelihood_of_subdata}.

\noindent \textbf{Step 3}. Hold $\boldeta_1, \dots,\boldeta_d$ and $R_{S_1},\dots, R_{S_n}$ fixed of their estimates obtained in steps 1 and 2, estimate $R_{S_i, S_j}$ for $i<j$ simultaneously according to condition labels by maximizing the log-likelihood in \cref{eq:log_likelihood_of_data}, through the approach of parameterization in \cref{sec:partition_with_multi}.
Note that the positive definiteness of $R$ needs to be guaranteed.

\noindent \textbf{Step 4}. If necessary, hold only $\boldeta_1, \dots,\boldeta_d$ fixed of their estimates obtained in step 1, and use the estimates of $R_{S_1},\dots, R_{S_n}$ and $R_{S_i, S_j}$ for $i<j$ in steps 2 and 3 as a starting point, update the estimates of $R_{S_1},\dots, R_{S_n}$ and $R_{S_i, S_j}$ for $i<j$ simultaneously by maximizing \cref{eq:log_likelihood_of_data}, under the constraint that $R$ is positive definite.

\medskip

The key {step is to estimate the correlation sub-matrices $R_{S_1},\dots,R_{S_n}$ for all
sub-processes before estimating the matrices $R_{S_i, S_j}$ that measure the dependence structure between the sub-processes.
The method is based on the property that all sub-processes in the partition follow VAR($k$) models under the constraint of closure under margins.
Otherwise, if a sub-process is not Markov or has Markov order not equal to $k$, the estimation of the correlation matrix of $k+1$ consecutive observations of the sub-process may be biased. The idea of closure under margins would be helpful especially in the situation of a high-dimensional data sets with inadequate sample size for unconstrained VAR models.
In this case, the set of parameters of $R_{S_i, S_j}$ for $i<j$ will be reduced significantly if the original time series are partitioned into sub-processes with much lower dimensions.
More importantly, by fitting each sub-process individually before estimating
the whole correlation matrix $R$, the problem of maximizing the likelihood
with a high-dimension input is divided into several lower-dimensional
optimization problems, and this reduces computational complexity.

\section{A numerical example and an empirical study}
\label{sec:examples}

In this section, a numerial example and an application are given.
Section \ref{sec:numerical} has a numerical example that includes
comparisons of coefficient matrices and covariance matrices of the
innovation vector for the different cases in Section \ref{sec:closure}.
Section \ref{sec:application} has an application to  a macro-economic data set.

\subsection{Numerical example of a bivariate VAR(2) model}
\label{sec:numerical}

We give a numerical example of a bivariate margin-closed VAR(2) model to discuss the behavior of the coefficient matrices under different condition labels.
Consider a 2-dimensional VAR(2) process $\cbra{\rbra{Z_{1,t},Z_{2,t}}}_{t>0}$ with the partition $\cbra{\cbra{1}, \cbra{2}}$.
Suppose the $Z_{1,t}$ and $Z_{2,t}$ have standard Gaussian margins, and the correlation matrices of $\rbra{Z_{1,t}, Z_{1,t-1}, Z_{1,t-2}}^T$ and $\rbra{Z_{2,t}, Z_{2,t-1}, Z_{2,t-2}}^T$ are
\[
    R_{\cbra{1}} =
    \begin{pmatrix}
    1 & -0.8 &   0.6\\
     -0.8 &   1 & -0.8\\
    0.6 & -0.8 &   1
    \end{pmatrix}\ \mbox{and}\
    R_{\cbra{2}} =
    \begin{pmatrix}
    1 &   0.6 &   0.5\\
    0.6 &   1 &   0.6\\
    0.5 &   0.6 &   1
    \end{pmatrix},
\]
respectively.
Note that both $R_{\cbra{1}}$ and $R_{\cbra{2}}$ are Toeplitz matrices so that the two univariate sub-processes are stationary AR(2).
The two univariate representations under the constraint that $\cbra{\rbra{Z_{1,t},Z_{2,t}}}_{t>0}$ is closed under margins with respect to $\cbra{\cbra{1}, \cbra{2}}$
are as follows:
\[
    Z_{1,t} &= -0.889 Z_{1,t-1} - 0.111 Z_{1,t-2} + \epsilon_{Z_1,t},\ \epsilon_{Z_1,t}\overset{\mathbf{i.i.d.}}{\sim}N(0, 0.356);\\
    Z_{2,t} &= 0.469 Z_{2,t-1} + 0.219 Z_{2,t-2} + \epsilon_{Z_2,t},\ \epsilon_{Z_2,t}\overset{\mathbf{i.i.d.}}{\sim}N(0, 0.609).
\]
To see the behavior of the coefficient matrices under different condition labels, suppose
$\corr\rbra{Z_{1,t}, Z_{2,t}} = 0.35$ in the cases of condition labels $(1,1)$ and $(2,2)$, $\corr\rbra{Z_{1,t-2}, Z_{2,t}} = 0.35$ in the case of condition labels $(1,2)$, and $\corr\rbra{Z_{1,t}, Z_{2,t-2}} = 0.35$ in the case of condition labels $(2,1)$.
Then the other cross-sectional dependence parameters between AR(2) sub-processes $\cbra{Z_{1,t}}_{t>0}$ and $\cbra{Z_{2,t}}_{t>0}$ can be derived following the procedure in Section \ref{sec:closure}.
It is easy to verify the positive definiteness of the correlation matrix $R$ in all cases. \cref{tb:VAR_parameters_of_2d_VAR(2)} shows the derived coefficient matrices
$\bPhi_{1}, \bPhi_{2}$
and covariance matrix $\bSigma_{\bepsilon}$ of the innovation vectors of VAR(2) process $\cbra{\rbra{Z_{1,t},Z_{2,t}}}_{t>0}$.

\begin{table}[H]
\small
\centering
\begin{tabular}{ccccc}
  \toprule
  Condition & & \multirow{2}{*}{$\bPhi_{1}$} & \multirow{2}{*}{$\bPhi_{2}$} & \multirow{2}{*}{$\bSigma_{\bepsilon}$} \\
  labels & & &  \\\cmidrule{1-1}\cmidrule{3-5}
  \rule{0pt}{20pt}
  $(1,1)$ & & $\begin{pmatrix}
             -0.889 & 0\\
               0 &   0.469
            \end{pmatrix}$ &
            $\begin{pmatrix}
             -0.111 &   0\\
               0 &   0.219
            \end{pmatrix}$ &
            $\begin{pmatrix}
               0.356 &   0.447\\
               0.447 &   0.609
            \end{pmatrix}$
            \\[10pt]
  \rule{0pt}{20pt}
  $(1,2)$ & & $\begin{pmatrix}
             -0.889 &   0\\
               0.778 &   0.469
            \end{pmatrix}$ &
            $\begin{pmatrix}
             -0.111 &   0\\
                   0.972 &   0.219
            \end{pmatrix}$ &
            $\begin{pmatrix}
               0.356 &   0.039\\
               0.039 &   0.269
            \end{pmatrix}$
            \\[10pt]
  \rule{0pt}{20pt}
  $(2,1)$ & & $\begin{pmatrix}
             -0.889 & -0.328\\
               0 &   0.469
            \end{pmatrix}$ &
            $\begin{pmatrix}
             -0.111 &   0.547\\
               0 &   0.219
            \end{pmatrix}$ &
            $\begin{pmatrix}
               0.164 & -0.077\\
             -0.077 &   0.609
            \end{pmatrix}$
            \\[10pt]
  \rule{0pt}{20pt}
  $(2,2)$ & & $\begin{pmatrix}
             -0.716 &   0.656\\
             -1.184 &   0.296
            \end{pmatrix}$ &
            $\begin{pmatrix}
               0.353 & -0.330\\
             -0.863 &   0.736
            \end{pmatrix}$ &
            $\begin{pmatrix}
               0.194 & -0.196\\
             -0.196 &   0.287
            \end{pmatrix}$
            \\[10pt]
  \bottomrule
\end{tabular}
\captionsetup{font=small}
\caption{The coefficient matrices and covariance matrix of the innovation vector of VAR(2) process $\cbra{\rbra{Z_{1,t},Z_{2,t}}}_{t>0}$ with cross-correlation set as 0.35, under different sufficient conditions of closure under margins with respect to partition $\{\{1\}, \{2\}\}$.}
\label{tb:VAR_parameters_of_2d_VAR(2)}
\end{table}

Similar to the results of \cref{ex:2_dimensional_VAR(1)}, the coefficient matrices of the VAR
process $\cbra{\rbra{Z_{1,t},Z_{2,t}}}_{t>0}$ are all diagonal if both sub-processes satisfy
Condition 1 while the non-diagonal coefficient matrices can be obtained in the other three cases.
However, the non-diagonal elements in the first row of the coefficient matrices are $0$ if the first sub-process satisfies Condition 1, and the non-diagonal elements in the second row are $0$
if the second sub-process satisfies Condition 1.
It corresponds to our previous conclusion that all coefficient matrices of the original VAR process will be diagonal in blocks if all sub-processes fulfill Condition 1, and all sub-processes should satisfy Condition 2 if all non-diagonal blocks of the coefficient matrices of the original VAR process are required to be non-zero.

\cref{tb:VAR_parameters_of_2d_VAR(2)} also shows different covariance matrices of the innovation vectors under different condition labels.
$\epsilon_{Z_1,t}$ and $\epsilon_{Z_2,t}$ are positively correlated for condition labels $(1,1)$ and $(1,2)$, while they are negatively correlated for condition labels $(2,1)$ and $(2,2)$.
The differences between the covariance matrices result from the same value of cross-correlation of $0.35$ but different interpretations of the fixed parameters under different condition labels.
Condition labels $(1,1)$ and $(2,2)$ refer to the fixed parameter as the contemporaneous cross-sectional dependence between two sub-processes while the fixed parameter is the cross-sectional dependence at lag $-2$ and $2$ between two sub-processes for condition labels $(1,2)$ and $(2,1)$, respectively.

The results from the fixed parameter in \cref{tb:value_of_free_par_2d_VAR(2)} are helpful to further understand the correlation structure of the innovation vectors under the constraint of closure under margins.
The corresponding coefficient matrices,  as well as the correlation matrices of the innovation vectors are presented in \cref{tb:VAR_parameters_of_2d_VAR(2)_2}.
\begin{table}[H]
\small
\centering
\begin{tabular}{cccc}
\toprule
Condition labels & & Fixed parameter                  & Value \\\cmidrule{1-1}\cmidrule{3-4}
$(1,1)$          & & $\corr\rbra{Z_{1,t}, Z_{2,t}}$   & $0.292$ \\
$(1,2)$          & & $\corr\rbra{Z_{1,t-2}, Z_{2,t}}$ & $0.464$ \\
$(2,1)$          & & $\corr\rbra{Z_{1,t}, Z_{2,t-1}}$ & $-0.459$ \\
$(2,2)$          & & $\corr\rbra{Z_{1,t}, Z_{2,t}}$   & $-0.346$ \\
\bottomrule
\end{tabular}
\captionsetup{font=footnotesize}
\caption{The values of the fixed parameters of VAR(2) process $\cbra{\rbra{Z_{1,t},Z_{2,t}}}_{t>0}$
under different sufficient conditions for closure under margins with respect to partition
$\{\{1\}, \{2\}\}$. The values were chosen to get similar correlations between components of the innovation vector.
}
\label{tb:value_of_free_par_2d_VAR(2)}
\end{table}

\begin{table}[H]
\small
\centering
\begin{tabular}{ccccc}
  \toprule
  Condition & & \multirow{2}{*}{$\bPhi_{1}$} & \multirow{2}{*}{$\bPhi_{2}$} & Correlation  \\
  labels & & & & matrix of $\bepsilon_t$\\\cmidrule{1-1}\cmidrule{3-5}
  \rule{0pt}{20pt}
  $(1,1)$ & & $\begin{pmatrix}
             -0.889 & 0\\
               0 &   0.469
            \end{pmatrix}$ &
            $\begin{pmatrix}
             -0.111 &   0\\
               0 &   0.219
            \end{pmatrix}$ &
            $\begin{pmatrix}
               1 &   0.801\\
               0.801 &   1
            \end{pmatrix}$
            \\[10pt]
  \rule{0pt}{20pt}
  $(1,2)$ & & $\begin{pmatrix}
             -0.889 &   0\\
               1.031 &   0.469
            \end{pmatrix}$ &
            $\begin{pmatrix}
             -0.111 &   0\\
               1.289 &   0.219
            \end{pmatrix}$ &
            $\begin{pmatrix}
               1 &   0.812\\
               0.812 &   1
            \end{pmatrix}$\\[10pt]
  \rule{0pt}{20pt}
  $(2,1)$ & & $\begin{pmatrix}
             -0.889 & 0.430\\
               0 &   0.469
            \end{pmatrix}$ &
            $\begin{pmatrix}
             -0.111 &   -0.717\\
               0 &   0.219
            \end{pmatrix}$ &
            $\begin{pmatrix}
               1 &   0.792\\
               0.792 &   1
            \end{pmatrix}$\\[10pt]
  \rule{0pt}{20pt}
  $(2,2)$ & & $\begin{pmatrix}
             -0.787 & -0.590\\
               1.080 &   0.367
            \end{pmatrix}$ &
            $\begin{pmatrix}
               0.243 &   0.246\\
               0.721 &   0.630
            \end{pmatrix}$ &
            $\begin{pmatrix}
               1 &   0.797\\
               0.797 &   1
            \end{pmatrix}$\\[10pt]
  \bottomrule
\end{tabular}
\captionsetup{font=footnotesize}
\caption{The coefficient and covariance matrices of the innovation vector of VAR(2) process $\cbra{\rbra{Z_{1,t},Z_{2,t}}}_{t>0}$ with values of fixed parameters in \cref{tb:value_of_free_par_2d_VAR(2)}, under different sufficient conditions of closure under margins with respect to partition $\{\{1\}, \{2\}\}$.}
\label{tb:VAR_parameters_of_2d_VAR(2)_2}
\end{table}

As shown in \cref{tb:VAR_parameters_of_2d_VAR(2)_2}, the values of the fixed parameters in
\cref{tb:value_of_free_par_2d_VAR(2)} can lead to similar correlations between $\epsilon_{Z_1,t}$
and $\epsilon_{Z_2,t}$ even though the non-diagonal entries in their coefficient matrices are distinct.
Therefore, by choosing appropriate values, even though the interpretation of the fixed parameter will vary with the condition label, they indeed can result in similar correlation matrices of the innovation vectors while give different types of coefficient matrices.
It shows the flexibility of the margin-closed VAR models.

\subsection{Application}
\label{sec:application}

This section illustrates the margin-closed VAR model on a trivariate multivariate time series from the FRED monthly database \citep{Mccracken2016fred}.
The database contains many macroeconomic variables with monthly frequency. All variables
have been transformed, possibly through differencing, so that
the assumption of stationarity in the resulting time series may be reasonable.

\subsubsection{Three variables from FRED monthly database}

For an illustration,
three variables are used:
the total consumer loans and leases outstanding (CLL), the real personal consumption expenditures (PCE), and the consumer price index (CPI).
The three variables were transformed to the second, the first, and the second difference of natural log, respectively.
All values are presented on the scale of percentages.
Based on plots,
approximate stationarity seems acceptable for the period from Mar.~1989 to Aug.~2001, inclusive; this corresponds to 150 consecutive months of date.
\cref{fig:time_series_plots} gives the plots of the transformed trivariate time series.

\begin{figure}[ht]
  \centering
  \includegraphics[scale=0.4]{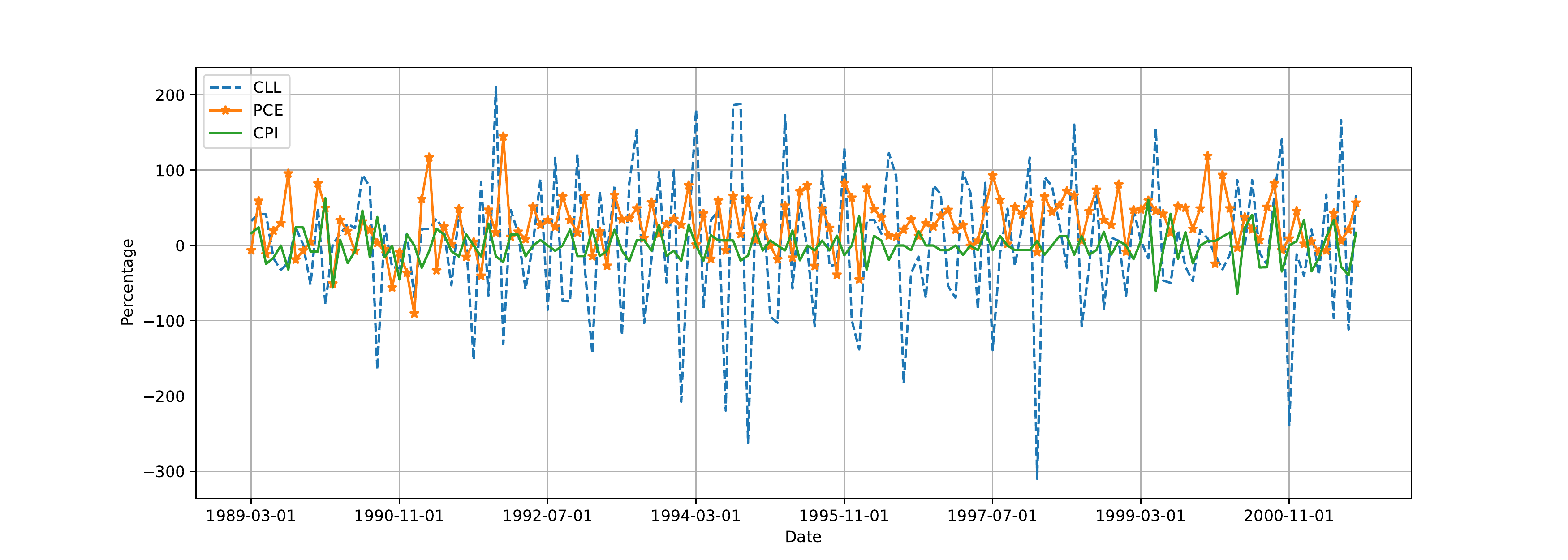}
  \captionsetup{font=footnotesize}
  \caption{The plots of transformed time series of the total consumer loans and leases outstanding (CLL), the real personal consumption expenditures (PCE), and the consumer price index (CPI), from Mar. 1989 to Aug. 2001.}
  \label{fig:time_series_plots}
\end{figure}

\subsubsection{Model comparisons}

We employ the dependence structure in margin-closed VAR model and the unrestricted VAR model as the multivariate copula and compare their performances in the above dataset.
To fit the model, we first determine the marginal distribution of each univariate component.
As the univariate margins may have heavy tails and skewness,
we fit them with both the Gaussian distribution and the
extended skew t-distribution \citep{Jones2003skew_t}.
The Akaike information criterion (AIC) values presented in \cref{tb:AIC_of_margins} indicate that the extended skew t-distribution is preferred for the variables CLL and CPI. The quantile-quantile plots (not included here) confirm a good agreement between the model and the data, and show that the extended skew t-distribution better captures the tails for these two variables.


\begin{table}[ht]
\small
\centering
\begin{tabular}{cccc}
\toprule
  & CLL  & PCE & CPI \\
\cmidrule{2-4}
Extended skew t-distribution & 403 & 128 & -40 \\
Gaussian distribution & 404 & 126 & -38 \\\bottomrule
\end{tabular}
\captionsetup{font=footnotesize}
\caption{The AIC values of extended skew t-distribution and Gaussian distribution, fitting the univariate margins of each univariate component of the trivariate time series.}
\label{tb:AIC_of_margins}
\end{table}

The extended skew t-distribution has two extra parameters $a$ and $b$, in addition to location and scale parameters in order to control the left and right tailweights.
Estimated univariate marginal parameters are shown in \cref{tb:parameters_of_margins}.
The fitted values indicate heavy tails,
relative to Gaussian, in both sides of the univariate densities of CLL and PCE.
This also indicates the need for non-Gaussian margins for the stationary joint distribution in \cref{eq:stationary_joint_cdf_by_copula}.

\begin{table}[ht]
\small
\centering
\begin{tabular}{ccccc}
\toprule                & Location  & Scale     & Left tailweight & Right tailweight \\
                      & parameter & parameter & parameter & parameter \\\cmidrule{2-5}
CLL                   &  0.850 & 0.791 & 5.739 & 9.344 \\
PCE                   & 0.281     & 0.363     & ---     & ---\\
CPI                   & -0.032    & 0.172     & 3.053      & 2.738 \\\bottomrule
\end{tabular}
\captionsetup{font=footnotesize}
\caption{The estimates of the model parameters for each univariate component of the trivariate time series. The CLL and CPI series are fitted using the extended skew t-distribution, while the PCE series is fitted using the Gaussian distribution.}
\label{tb:parameters_of_margins}
\end{table}

For the dependence structure in the stationary joint distribution in \cref{eq:stationary_joint_cdf_by_copula}, the multivariate Gaussian copula is a good choice because of the limited sample size.
It is equivalent to a latent VAR process in the multivariate time series model.
We fit the margin-closed VAR model and the unrestricted VAR model with the pseudo-observations of the latent VAR process obtained by transforming all univariate components to be standard normal,
then compare the margin-closed and unrestricted time series models
that include fitting the univariate margins.
To fit the margin-closed model, other hyperparameters should be specified including the partitions and the condition labels.
Since there are only three univariate components and the simpler interpretations of non-diagonal coefficient matrices are of our main interests,
we consider the simple case of the partition $\cbra{\{1\}, \{2\}, \{3\}}$ and the condition label $(2,2,2)$.
To determine the Markov order $k$ of the model, we compare the AIC values of two time series models with different Markov orders.
\cref{tb:Num_of_pars_and_AIC} shows the number of parameters and AIC values of the two models with Markov order from 1 to 5.

Because of the constraint of closure under margins,
the margin-closed model has significantly fewer number of parameters than the unrestricted model,
and the reduction in the parameter set would be even larger in the
case of higher dimension of data sets and higher Markov orders.
It can be noticed that the AIC values of margin-closed model are similar to
that of the unrestricted model in the case of Markov order 1,
and smaller in all case of higher Markov orders.
As Markov order 2 leads to the minimum AIC values for the margin-closed
and unrestricted models,
the margin-closed model with Markov order 2 is preferred in view of model parsimony.

\begin{table}[ht]
\small
\centering
\begin{tabular}{ccccccccc}
\toprule
& & \multicolumn{3}{c}{Margin-closed model} &  & \multicolumn{3}{c}{Unrestricted model} \\\cmidrule{3-5}\cmidrule{7-9}
Markov order & & No. of parameters &  & AIC & & No. of parameters &  & AIC
\\\cmidrule{1-1}\cmidrule{3-3}\cmidrule{5-5}\cmidrule{7-7}\cmidrule{9-9}
$k=1$ & & 16 & & 447 & & 22 & & 446 \\
$k=2$ & & 19 & & 397 & & 31 & & 409 \\
$k=3$ & & 22 & & 405 & & 40 & & 420 \\
$k=4$ & & 25 & & 405 & & 49 & & 428 \\
$k=5$ & & 28 & & 399 & & 58 & & 429 \\
\bottomrule
\end{tabular}
\captionsetup{font=footnotesize}
\caption{The AIC values of two models with different Markov orders.
The margin-closed model has $k$ serial dependence parameters for each
variable and 3 contemporaneous dependence parameters.
The unrestricted VAR(2) model has in addition $6\times k$
cross-correlations at lags $\pm1,\dots,\pm k$.}
\label{tb:Num_of_pars_and_AIC}
\end{table}

\begin{table}[h]
\small
\centering
\begin{tabular}{ccc}
\toprule
 & Margin-closed latent & Unrestricted latent \\
 &  VAR(2) process   & latent VAR(2) process \\\cmidrule{2-3}
 \rule{0pt}{25pt}
$\bPhi_1$ & $\begin{pmatrix}
             -0.538 &   0.033 & -0.002\\
             -0.067 & -0.109 & -0.045\\
             -0.006 &   0.017 & -0.551
            \end{pmatrix}$ &
            $\begin{pmatrix}
             -0.535 &   0.058 &   0.093\\
             -0.035 & -0.098 & -0.116\\
             -0.047 &   0.157 & -0.540
            \end{pmatrix}$\\
            [15pt]
  \rule{0pt}{25pt}
$\bPhi_2$ & $\begin{pmatrix}
             -0.376 &   0.040 &   0.007\\
             -0.047 & -0.003 & -0.035\\
             -0.015 &   0.029 & -0.455
            \end{pmatrix}$ &
            $\begin{pmatrix}
             -0.370 & -0.016 & -0.005\\
             -0.057 &   0.024 &   0.062\\
             -0.045 &   0.054 & -0.435
            \end{pmatrix}$\\
            [15pt]
  \rule{0pt}{25pt}
$\bSigma_{\bepsilon}$ & $\begin{pmatrix}
               0.729 &   0.084 & -0.096\\
               0.084 &   0.981 &   0.043\\
             -0.096 &   0.043 &   0.681
            \end{pmatrix}$ &
            $\begin{pmatrix}
               0.713 &   0.120 & -0.109\\
               0.120 &   0.965 & -0.017\\
             -0.109 & -0.017 &   0.663
            \end{pmatrix}$\\
            [15pt]
\bottomrule
\end{tabular}
\captionsetup{font=footnotesize}
\caption{The coefficient matrices and covariance matrices of the innovation vectors of fitted latent margin-closed and unrestricted VAR(2) processes for the pseudo-observations.
The partition and condition label of the margin-closed VAR(2) model are $\{\{1\}, \{2\}, \{3\}\}$ and $(2,2,2)$.
$\bPhi_i$ indicates the coefficient matrix of observation at lag $i$ for $1\leq i\leq 2$, $\bSigma_{\bepsilon}$ is the covariance matrix of the innovation vectors.}
\label{tb:parameters_of_mc_VAR and_VAR}
\end{table}

The fitted parameters of the latent margin-closed and unrestricted VAR(2) models are presented in the forms of the coefficient matrices and the covariance matrices of the innovation vectors in \cref{tb:parameters_of_mc_VAR and_VAR}.
It is seen that most of the estimated parameters of the margin-closed model are near to the corresponding estimates of the unrestricted model, especially the entries in the diagonals of all fitted matrices.
As for the covariance matrices of the innovation vectors, except for the closeness of the diagonals, the sign of most of the correlation coefficients in non-diagonal parts are the same except for the correlation coefficient between the innovation terms of the second and third univariate sub-processes.
That may be due to the fact that its estimates in both models are close to 0. Moreover, the Frobenius norm of the difference between estimated correlation matrices of the innovation vectors of the two models is only $0.105$; this verifies the closeness of the dependence structure of the innovation vectors of the two models. For the latent margin-closed VAR(2) model, the ACF values and multivariate Portmanteau (Ljung-Box) test of residuals also indicate the adequacy of the model.

\section{Discussion}
\label{sec:conclusion}

The sufficient conditions for closure under margins mean that all sub-processes of a given partition of a VAR($k$) process can be
AR($k$) or VAR($k$), with non-diagonal coefficient matrices of the VAR($k$) process.
By employing non-Gaussian univariate margins for each univariate component and multivariate Gaussian copulas of stationary joint distributions of margin-closed VAR models, the margin-closed time series models have great flexibility in fitting the high-dimensional time series.

The numerical examples of the margin-closed model show its capacity to give the non-diagonal coefficient matrices without changing
the correlation structure of the innovation vectors.
The margin-closed model is also applied to a trivariate macro-economic data set, and the results indicate better performance of the
margin-closed model compared with the unrestricted model.

Code will be made available for checking
conditions for positive definiteness in cases similar to
\cref{ex:2_dimensional_VAR(1)} to \cref{ex:3_dimensional_VAR(1)_2}.
This will also help with the estimation steps in Section \ref{sec:estimation}.

\section*{Acknowledgements}
We are grateful to the Editor and two referees for helpful suggestions that improved presentation and clarity of the paper. The research was funded in part by grants GR009581 and GR010293 from the Natural Sciences and Engineering Research Council of Canada and in part by the UBC-Scotiabank Risk Analytics Initiative.

\bibliographystyle{apalike}

\appendix

\section{Proofs}

The following lemma is needed for proving \cref{thm:sufficient_condition}.

\begin{lemma}\label{lemma:lemma_for_proving_thm}
Let $\bA_1$, $\bA_2$, $\bB$, and $\bW$ be four multivariate Gaussian distributed random vectors. If $\bB$ is independent of $(\bA_1^T,\bA_2^T)^T$ given $\bW$, then $\bB$ is independent of $\bA_2$ given $\bA_1$ and $\bW$.
\end{lemma}

\begin{proofof}{\cref{lemma:lemma_for_proving_thm}}
Let $\bSigma_{\bA_1,\bA_2|\bW}$, $\bSigma_{\bA_1,\bB|\bW}$, $\bSigma_{\bA_1,\bB|\bW}$ be covariance matrices between $\bA_1$ and $\bA_2$, between $\bA_1$ and $\bB$, and between $\bA_2$ and $\bB$, respectively, conditional on $\bW$.
Also let $\bSigma_{\bA_1|\bW}$ be covariance matrix of $\bA_1$ given $\bW$.
Then the block diagonal form of the covariance matrix of $(\bA_1^T,\bA_2^T,\bB^T)^T$
given $\bW$ implies that $\bSigma_{\bA_1,\bB|\bW}=\bm{0}$ and $\bSigma_{\bA_2,\bB|\bW}=\bm{0}$.
From the conditional covariance of multivariate Gaussian random vectors,
  \[
    \bSigma_{\bA_2,\bB|\bA_1,\bW} &= \bSigma_{\bA_2,\bB|\bW}- \bSigma_{\bA_2,\bA_1|\bW}^T\bSigma_{\bA_1|\bW}^{-1}\bSigma_{\bA_1,\bB|\bW}\\
    &= \bm{0} - \bSigma_{\bA_2,\bA_1|\bW}\bSigma_{\bA_1|\bW}^{-1}\times\bm{0}
    = \bm{0},
  \]
  which indicates the independence between $\bA_2$ and $\bB$ given $\bA_1$ and $\bW$.
\end{proofof}

The proof of \cref{thm:sufficient_condition} is given next.

\begin{proofof}{\cref{thm:sufficient_condition}}

The main idea is to use the induction to show that the required condition in \cref{eq:condition_of_being_mc} is satisfied for all $l\geq k+1$ once it is satisfied for $l=k+1$. Let
\[
    \bV_{1|t,l} &= \bZ_{S_i,t},\quad \bV_{2|t,l} =\rbra{\bZ_{-S_i,t-1}^T,\dots,\bZ_{-S_i,t-l+1}^T}^T,\\
    \bV_{3|t,l} &= \bZ_{S_i,t-l},\quad \bV_{4|t,l} =\rbra{\bZ_{S_i,t-1}^T,\dots,\bZ_{S_i,t-l+1}^T}^T.
\]
Let $\bGamma_{i,j|l}$ denote the covariance matrix between $\bV_{i|t,l}$ and $\bV_{j|t,l}$ given $\bV_{4|t,l}$ for $1\leq i,j\leq 3$.
The covariance matrix of $\rbra{\bV_{1|t,l}^T, \bV_{2|t,l}^T, \bV_{3|t,l}^T}^T$ given $\bV_{4|t,l}$ can be written as
\begin{equation}\label{eq:cov_Gamma123}
  \begin{pmatrix}
  \bGamma_{1,1|l} & \bGamma_{1,2|l} & \bGamma_{1,3|l} \\
  \bGamma_{1,2|l}^T & \bGamma_{2,2|l} & \bGamma_{2,3|l} \\
  \bGamma_{1,3|l}^T & \bGamma_{2,3|l}^T & \bGamma_{3,3|l}
\end{pmatrix}.
\end{equation}

Suppose \cref{eq:condition_of_being_mc} is satisfied for a fixed $l\geq k+1$, so that $\bGamma_{1,2|l}=\bm{0}$ or $\bGamma_{2,3|l}=\bm{0}$.
To verify the induction hypothesis, it is necessary to show that $\bGamma_{1,2|l+1}=\bm{0}$ or $\bGamma_{2,3|l+1}=\bm{0}$.
Since $\cbra{\bZ_t}_{t>0}$ is Markov of order $k$, $\sbra{\bV_{1|t,l}\bot \bV_{3|t,l}}|\bV_{2|t,l}, \bV_{4|t,l}$.
By \cref{lemma:Gaussian_partial_corr}, $\bGamma_{1,2|l}=\bm{0}$ or $\bGamma_{2,3|l}=\bm{0}$ implies $\bGamma_{1,3|l}=\bm{0}$.
We discuss Case (i) of $\bGamma_{1,2|l}=\bm{0}$ and Case (ii) of $\bGamma_{2,3|l}=\bm{0}$ separately.

\medskip

\noindent
\textbf{Case (i)}. In this case $\bGamma_{1,2|l}=\bm{0}$ and $\bV_{4|t,l+1}=\rbra{\bV_{4|t,l}^T, \bV_{3|t,l}^T}^T$. According to the definition, it follows that
\begin{align}\label{eq:Gamma12_1}
  &\cov\rbra{\bV_{1|t,l}, \bV_{2|t,l} \big| \bV_{4|t,l+1}}
    =\cov\rbra{\bV_{1|t,l}, \bV_{2|t,l} \big| \bV_{4|t,l}, \bV_{3|t,l}} \\
    =& \bSigma_{\bV_{1|t,l}, \bV_{2|t,l} \big|\bV_{4|t,l}}
    - \bSigma_{\bV_{1|t,l}, \bV_{3|t,l}\big|\bV_{4|t,l}}
     \bSigma^{-1}_{\bV_{3|t,l}, \bV_{3|t,l}\big|\bV_{4|t,l}}
     \bSigma_{\bV_{3|t,l}, \bV_{2|t,l}\big|\bV_{4|t,l}} \\
    =& \bGamma_{1,2|l} - \bGamma_{1,3|l}\bGamma_{3,3|l}^{-1}\bGamma_{2,3|l}^T
    = \bm{0}.
\end{align}
Since $\cbra{\bZ_t}_{t>0}$ has Markov order $k$, then $\bV_{1|t,l} = \bZ_{S_i,t}$ is independent of $\bZ_{t-l}=(\bV_{3|t,l}^T, \bZ_{-S_i,t-l}^T)^T$ given $(\bV_{2|t,l}, \bV_{4|t,l})$.
By \cref{lemma:lemma_for_proving_thm},
\begin{equation}\label{eq:Gamma12_2}
  \cov\rbra{\bV_{1|t,l}, \bZ_{-S_i,t-l}\big|\bV_{2|t,l}, \bV_{4|t,l}, \bV_{3|t,l}}=\bm{0}
\end{equation}
Combining \cref{eq:Gamma12_1} and \cref{eq:Gamma12_2} leads to
\[
   &\cov\rbra{\bV_{1|t,l}, \bZ_{-S_i,t-l}\big|\bV_{4|t,l+1}}
    =\bSigma_{\bV_{1|t,l}, \bZ_{-S_i,t-l}\big|\bV_{4|t,l}, \bV_{3|t,l}}\\
    =& \bSigma_{\bV_{1|t,l}, \bZ_{-S_i,t-l}\big|\bV_{2|t,l}, \bV_{4|t,l}, \bV_{3|t,l}}
    - \bSigma_{\bV_{1|t,l}, \bV_{2|t,l}\big|\bV_{4|t,l}, \bV_{3|t,l}}
     \bSigma^{-1}_{\bV_{2|t,l},\bV_{2|t,l}\big|\bV_{4|t,l}, \bV_{3|t,l}}
     \bSigma_{\bV_{2|t,l}, \bZ_{-S_i,t-l}\big|\bV_{4|t,l}, \bV_{3|t,l}}\\
    =& \bm{0}.
\]
Since $\bV_{1|t,l+1}=\bV_{1|t,l}$ and $\Gamma_{1m2|l}=\bm{0}$,
\[
    &\bGamma_{1,2|l+1}
    =\cov\rbra{\bV_{1|t,l+1}, \bV_{2|t,l+1} | \bV_{4|t,l+1}}\\
    =&\cov\rbra{\bV_{1|t,l+1}, \rbra{\bV_{2|t}^T, \bZ_{-S_i,t-l}^T}^T \Big| \bV_{4|t,l+1}}=\bm{0}.
\]

\noindent
\textbf{Case (ii)}. In this case $\bGamma_{2,3|l}=\bm{0}$ and
$\bV_{4|t,l+1}=\rbra{\bV_{1|t-1,l}^T, \bV_{4|t-1,l}^T}^T$.
The stationarity of $\{\bZ_t\}_{t>0}$ implies that the matrix in
\cref{eq:cov_Gamma123} is also the covariance matrix of $\rbra{\bV_{1|t-1,l}^T, \bV_{2|t-1,l}^T, \bV_{3|t-1,l}^T}^T$ given $\bV_{4|t-1,l}$.
Therefore,
\begin{align}\label{eq:Gamma23_1}
  &\cov\rbra{\bV_{2|t-1,l}, \bV_{3|t-1,l} \big| \bV_{4|t,l+1}}
   =\cov\rbra{\bV_{2|t-1,l}, \bV_{3|t-1,l} \big| \bV_{1|t-1,l}, \bV_{4|t-1,l}}\\
    =& \bSigma_{\bV_{2|t-1,l}, \bV_{3|t-1,l} \big|\bV_{4|t-1,l}}
     - \bSigma_{\bV_{2|t-1,l}, \bV_{1|t-1,l}\big|\bV_{4|t-1,l}}
     \bSigma^{-1}_{\bV_{1|t-1,l}, \bV_{1|t-1,l}\big|\bV_{4|t-1,l}}
     \bSigma_{\bV_{1|t-1,l}, \bV_{3|t-1,l}\big|\bV_{4|t-1,l}}\\
    =& \bGamma_{2,3|l} - \bGamma_{2,1|l}^T\bGamma_{1,1|l}^{-1}\bGamma_{1,3|l}
    =\bm{0}.
\end{align}
Since $\cbra{\bZ_t}_{t>0}$ has Markov order $k$, $\bV_{3|t-1,l} = \bZ_{S_i,t-l-1}$ is independent of $\bZ_{t-1}=(\bV_{1|t-1,l}^T, \bZ_{-S_i,t-1}^T)^T$ given $\bV_{2|t-1,l}, \bV_{4|t-1,l}$.
By \cref{lemma:lemma_for_proving_thm},
\begin{equation}\label{eq:Gamma23_2}
  \cov\rbra{\bZ_{-S_i,t-1}, \bV_{3|t-1,l}\big|\bV_{2|t-1,l}, \bV_{4|t-1,l}, \bV_{1|t-1,l}}=\bm{0}.
\end{equation}
Combining \cref{eq:Gamma23_1} and \cref{eq:Gamma23_2} leads to
\[
    &\cov\rbra{\bZ_{-S_i,t-1}, \bV_{3|t-1,l}\big|\bV_{4|t,l+1}}
    =\cov\rbra{\bZ_{-S_i,t-1}, \bV_{3|t-1,l}\big|\bV_{1|t-1,l}, \bV_{4|t-1,l}}\\
    =& \bSigma_{\bZ_{-S_i,t-1}, \bV_{3|t-1,l}\big|\bV_{2|t-1,l}, \bV_{1|t-1,l}, \bV_{4|t-1,l}} - \\
    &\bSigma_{\bZ_{-S_i,t-1}, \bV_{2|t-1,l}\big|\bV_{1|t-1,l}, \bV_{4|t-1,l}}
    \bSigma^{-1}_{\bV_{2|t-1,l},\bV_{2|t-1,l}\big|\bV_{1|t-1,l}, \bV_{4|t-1,l}}
    \bSigma_{\bV_{2|t-1,l}, \bV_{3|t-1,l}\big|\bV_{1|t-1,l}, \bV_{4|t-1,l}}\\
    =& \bm{0}
\]
It leads to
\[
    &\bGamma_{2,3|l+1}
    =\cov\rbra{\bV_{2|t,l+1}, \bV_{3|t,l+1} | \bV_{4|t,l+1}}\\
    =&\cov\rbra{\rbra{\bZ_{-S_i,t-1}^T,\bV_{2|t-1,l}^T}^T, \bV_{3|t-1,l} | \bV_{4|t,l+1}}=\bm{0}.
\]

Combining cases (i) and (ii), $\bGamma_{1,2|l+1}=\bm{0}$ or $\bGamma_{2,3|l+1}=\bm{0}$ and the induction hypothesis is verified.
Hence $\bGamma_{1,2|l}=\bm{0}$ or $\bGamma_{2,3|l}=\bm{0}$ for all $l\geq k+1$, which means that $\cbra{\bZ_{S_i,t}}_{t>0}$ follows a VAR($k$) model.
\end{proofof}

\section{Deriving linear systems}\label{appd: derive_linear_systems}

Let $\bA=\bZ_{S_1,t}$, $\bB=\rbra{\bZ_{S_2,t-1}^T,\dots,\bZ_{S_2,t-k}^T}^T$, $\bV=\rbra{\bZ_{S_1,t-1}^T,\dots,\bZ_{S_1,t-k}^T}^T$.
According to the conditional variance of Gaussian random vectors,
Condition 1 in \cref{eq:thm3_for_S1S2} for $\cbra{\bZ_{S_1,t}}_{t>0}$ can be expressed as:
\[ \bSigma_{\bA, \bB|\bV} = \bSigma_{\bA, \bB} - \bSigma_{\bA, \bV}
  \bSigma^{-1}_{\bV} \bSigma_{\bV, \bB}=\bm{0},\]
or
\begin{equation}\label{eq:condition1_for S1}
 \begin{pmatrix}
  \bSigma_{12, 1} & \bSigma_{12, 2} & \cdots & \bSigma_{12, k} \end{pmatrix} -
 \begin{pmatrix}
   \bPhi_{1,1} &\bPhi_{1,2} &\cdots & \bPhi_{1,k} \end{pmatrix}
 \begin{pmatrix}
  \bSigma_{12,0}   & \bSigma_{12,1}   & \cdots & \bSigma_{12,k-1} \\
  \bSigma_{12,-1}  & \bSigma_{12,0}   & \cdots & \bSigma_{12,k-2} \\
  \vdots           & \vdots           & \ddots & \vdots \\
  \bSigma_{12,1-k} & \bSigma_{12,2-k} & \cdots & \bSigma_{12,0}
 \end{pmatrix} = \bm{0},
\end{equation}
where $\bPhi_{1,1},\dots,\bPhi_{1,k}$ are defined as
\begin{equation}
\begin{pmatrix}
   \bPhi_{1,1}^T \\
   \bPhi_{1,2}^T \\
   \vdots \\
   \bPhi_{1,k}^T
 \end{pmatrix}^T =
 \begin{pmatrix}
  \bSigma_{11,1}^T \\
  \bSigma_{11,2}^T \\
  \vdots \\
  \bSigma_{11,k}^T
 \end{pmatrix}^T
 \begin{pmatrix}
  \bSigma_{11,0}     & \bSigma_{11,1}     & \cdots & \bSigma_{11,k-1} \\
  \bSigma_{11,1}^T   & \bSigma_{11,0}     & \cdots & \bSigma_{11,k-2} \\
  \vdots             & \vdots             & \ddots & \vdots \\
  \bSigma_{11,k-1}^T & \bSigma_{11,k-2}^T & \cdots & \bSigma_{11,0}
 \end{pmatrix}^{-1}.
\end{equation}

Following the same idea, now with $\bA=\bZ_{S_1,t-k-1}$
and $\bB$ and $\bV$ as before, Condition 2 is equivalent to:
\begin{equation}\label{eq:condition2_for S1}
 \begin{pmatrix}
  \bSigma_{12, -k} &  \cdots & \bSigma_{12, -1} \end{pmatrix} -
 \begin{pmatrix}
   \bPsi_{1,1} & \cdots & \bPsi_{1,k} \end{pmatrix}
 \begin{pmatrix}
  \bSigma_{12,0}   & \bSigma_{12,1}   & \cdots & \bSigma_{12,k-1} \\
  \bSigma_{12,-1}  & \bSigma_{12,0}   & \cdots & \bSigma_{12,k-2} \\
  \vdots           & \vdots           & \ddots & \vdots \\
  \bSigma_{12,1-k} & \bSigma_{12,2-k} & \cdots & \bSigma_{12,0}
 \end{pmatrix} = \bm{0},
\end{equation}
where $\bPsi_{1,1},\dots,\bPsi_{1,k}$ are defined as
\begin{equation}
\begin{pmatrix}
   \bPsi_{1,1}^T \\
   \bPsi_{1,2}^T \\
   \vdots \\
   \bPsi_{1,k}^T
 \end{pmatrix}^T =
 \begin{pmatrix}
  \bSigma_{11,-k}^T \\
  \bSigma_{11,1-k}^T \\
  \vdots \\
  \bSigma_{11,-1}^T
 \end{pmatrix}^T
 \begin{pmatrix}
  \bSigma_{11,0}     & \bSigma_{11,1}     & \cdots & \bSigma_{11,k-1} \\
  \bSigma_{11,1}^T   & \bSigma_{11,0}     & \cdots & \bSigma_{11,k-2} \\
  \vdots             & \vdots             & \ddots & \vdots \\
  \bSigma_{11,k-1}^T & \bSigma_{11,k-2}^T & \cdots & \bSigma_{11,0}
 \end{pmatrix}^{-1}.
\end{equation}

Let $D_1 = \rbra{ \bSigma_{12, -k}^T, \dots, \bSigma_{12,-1}^T, \bSigma_{12,0}^T, \bSigma_{12,1}^T, \dots, \bSigma_{12,k}^T }^T$. Then the two conditions in \cref{eq:thm3_for_S1S2} can be rewritten as two linear systems.
Condition 1 for $\cbra{\bZ_{S_1,t}}_{t>0}$ is:
\begin{equation}
 \begin{pmatrix}
  0      & \bPhi_{1,k} & \dots       & \bPhi_{1,2} & \bPhi_{1,1} & -I_{d_1}    & 0          & \cdots      & 0\\
  0      & 0           & \bPhi_{1,k} & \dots       & \bPhi_{1,2} & \bPhi_{1,1} & -I_{d_1}  & \cdots      & 0 \\
  \vdots & \ddots      & \ddots      & \ddots      & \ddots      & \ddots      &   \ddots    & \ddots      & \vdots \\
  0      & \cdots      & 0           & 0           & \bPhi_{1,k} & \cdots      & \bPhi_{1,2} & \bPhi_{1,1} & -I_{d_1}
 \end{pmatrix} D_1 =: G_1D_1 = \bm{0}.
\end{equation}
Condition 2 for $\cbra{\bZ_{S_1,t}}_{t>0}$ is:
\begin{equation}
 \begin{pmatrix}
  -I_{d_1} & \bPsi_{1,k} & \dots       & \bPsi_{1,2} & \bPsi_{1,1} & 0           & 0           & \cdots      & 0\\
  0          & -I_{d_1}  & \bPsi_{1,k} & \dots       & \bPsi_{1,2} & \bPsi_{1,1} & 0           & \cdots      & 0 \\
  \vdots     & \ddots      & \ddots      & \ddots      & \ddots      & \ddots      &   \ddots    & \ddots      & \vdots \\
  0          & \cdots      & 0           & -I_{d_1}  & \bPsi_{1,k} & \cdots      & \bPsi_{1,2} & \bPsi_{1,1} & 0
 \end{pmatrix} D_1 =: H_1D_1 =\bm{0}.
\end{equation}
It gives \cref{eq: two_linear_system1}. The reader can verify these for $k=1$ and $k=2$, and then generalize.

Similarly, let
\begin{multline}
  D_2 = \rbra{ \bSigma_{21, -k}^T, \dots, \bSigma_{21,-1}^T, \bSigma_{21,0}^T, \bSigma_{21,1}^T, \dots, \bSigma_{21,k}^T }^T = \\\rbra{ \bSigma_{12, k}, \dots, \bSigma_{12,1}, \bSigma_{12,0}, \bSigma_{12,-1}, \dots, \bSigma_{12,-k} }^T.
\end{multline}

Then, Conditions 1 and 2 for $\cbra{\bZ_{S_2,t}}_{t>0}$ can be rewritten as
the following two linear systems.
Condition 1 for $\cbra{\bZ_{S_2,t}}_{t>0}$ is:
\begin{equation}\label{eq: condition_S2_1}
\begin{pmatrix}
  0      & \bPhi_{2,k} & \dots       & \bPhi_{2,2} & \bPhi_{2,1} & -I_{d_2}  & 0           & \cdots      & 0\\
  0      & 0           & \bPhi_{2,k} & \dots       & \bPhi_{2,2} & \bPhi_{2,1} & -I_{d_2}  & \cdots      & 0 \\
  \vdots & \ddots      & \ddots      & \ddots      & \ddots      & \ddots      &   \ddots    & \ddots      & \vdots \\
  0      & \cdots      & 0           & 0           & \bPhi_{2,k} & \cdots      & \bPhi_{2,2} & \bPhi_{2,1} & -I_{d_2}
 \end{pmatrix} D_2 =: G_2D_2 = \bm{0}.
\end{equation}
Condition 2 for $\cbra{\bZ_{S_2,t}}_{t>0}$ is:
\begin{equation}\label{eq: condition_S2_2}
\begin{pmatrix}
  -I_{d_2} & \bPsi_{2,k} & \dots       & \bPsi_{2,2} & \bPsi_{2,1} & 0           & 0           & \cdots      & 0\\
  0          & -I_{d_2}  & \bPsi_{2,k} & \dots       & \bPsi_{2,2} & \bPsi_{2,1} & 0           & \cdots      & 0 \\
  \vdots     & \ddots      & \ddots      & \ddots      & \ddots      & \ddots      &   \ddots    & \ddots      & \vdots \\
  0          & \cdots      & 0           & -I_{d_2}  & \bPsi_{2,k} & \cdots      & \bPsi_{2,2} & \bPsi_{2,1} & 0
 \end{pmatrix} D_2 =:H_2D_2 = \bm{0}.
\end{equation}
where
\begin{align}
\begin{pmatrix}
  \bPhi_{2,1}^T \\
  \bPhi_{2,2}^T \\
  \vdots \\
  \bPhi_{2,k}^T
 \end{pmatrix}^T &=
 \begin{pmatrix}
  \bSigma_{22,1}^T \\
  \bSigma_{22,2}^T \\
  \vdots \\
  \bSigma_{22,k}^T
 \end{pmatrix}^T
 \begin{pmatrix}
  \bSigma_{22,0}     & \bSigma_{22,1}     & \cdots & \bSigma_{22,k-1} \\
  \bSigma_{22,1}^T   & \bSigma_{22,0}     & \cdots & \bSigma_{22,k-2} \\
  \vdots             & \vdots             & \ddots & \vdots \\
  \bSigma_{22,k-1}^T & \bSigma_{22,k-2}^T & \cdots & \bSigma_{22,0}
 \end{pmatrix}^{-1}\\
\mbox{and}
 \begin{pmatrix}
   \bPsi_{2,1}^T \\
   \bPsi_{2,2}^T \\
   \vdots \\
   \bPsi_{2,k}^T
 \end{pmatrix}^T &=
 \begin{pmatrix}
  \bSigma_{22,-k}^T \\
  \bSigma_{22,1-k}^T \\
  \vdots \\
  \bSigma_{22,-1}^T
 \end{pmatrix}^T
 \begin{pmatrix}
  \bSigma_{22,0}     & \bSigma_{22,1}     & \cdots & \bSigma_{22,k-1} \\
  \bSigma_{22,1}^T   & \bSigma_{22,0}     & \cdots & \bSigma_{22,k-2} \\
  \vdots             & \vdots             & \ddots & \vdots \\
  \bSigma_{22,k-1}^T & \bSigma_{22,k-2}^T & \cdots & \bSigma_{22,0}
 \end{pmatrix}^{-1}.
\end{align}

Let $J_{2k+1}$ is the $(2k+1)$-dimensional exchange matrix whose
elements in the anti-diagonal (or oppositie diagonal) are 1 and all other elements are zero,
and set $L_2 = J_{2k+1} \otimes I_{d_2}$.
It follows that
\[L_2D_{2} = L_2^{-1}D_{2} = \rbra{ \bSigma_{12, -k}, \dots, \bSigma_{12,-1}, \bSigma_{12,0}, \bSigma_{12,1}, \dots, \bSigma_{12,k} }^T.\]
Then \cref{eq: condition_S2_1} and \cref{eq: condition_S2_2}
can be written as
\begin{equation}
 \begin{aligned}
\mbox{Condition 1: } &
 (G_2L_2)(L_2D_2) = \bm{0};\\
\mbox{Condition 2: } &
 (H_2L_2)(L_2D_2) = \bm{0}.
\end{aligned}
\end{equation}

\section{Solving equations in {vec} notation}

\subsection{Two sub-processes}\label{appd:two_sub_processes}

The $\mbox{vec}\rbra{\cdot}$ is useful
to solve
\cref{eq: two_linear_system1} to \cref{eq: two_linear_system2}
for $d_1,d_2\ge2$ and $k\ge2$, especially when the components of
$G_1,H_1,G_2,H_2$ are not scalars.
In this Appendix, some details are shown.

Let $G_{1,:i}$ denote the submatrix of the $i$-th block column of $G_1$ and $G_{1,:-i}$ denote the submatrix obtained by removing the $i$-th block column of $G_1$.
Similarly, define $D_{1,-i:}$ as the submatrix obtained by removing the $i$-th block row of $D_{1}$. Because the blocks in $D_1$ and $L_2D_2$ are corresponding transposes,
the combination of the two systems can be solved (numerically) by
vectorizing the elements of $\bSigma_{12,j}$ for $j=-k,\ldots,k$.

\medskip

\noindent
\textbf{Case 1}.
When both $\cbra{\bZ_{S_1,t}}_{t>0}$ and $\cbra{\bZ_{S_2,t}}_{t>0}$ adopt Condition 1,
Let $k_1=k+1$ be the chosen column to fix for $\bSigma_{12,0}$.
It follows that
\[
    G_{1,:-k_1}D_{1,-k_1:} = G_{1,:k_1}\bSigma_{12,0}\ \mbox{and}\ (G_2L_2)_{:-k_1}(L_2D_2)_{-k_1:} = (G_2L_2)_{:k_1}\bSigma_{12,0}^T = G_{2,:k_1}\bSigma_{12,0}^T,
\]
which means $\bSigma_{12,0}$ is selected as the fixed parameter when we solve the equations.
In the above,\\ $\bSigma_{12,-k},\ldots,\bSigma_{12,-1},
\bSigma_{12,1},\ldots,\bSigma_{12,k}$ are stacked in $D_{1,-k_1:}$,
and
$\bSigma^T_{12,-k},\ldots,\bSigma^T_{12,-1},
\bSigma^T_{12,1},\ldots,\bSigma^T_{12,k}$ are stacked in $(L_2D_2)_{-k_1:}$.

By applying the $\mbox{vec}\rbra{\cdot}$ and using the formula that $\mbox{vec}(AB)=(A\otimes I_n)\mbox{vec}(B^T)$ where $A$ is $m\times p$ and $B$ is $p\times n$, the above equations lead to
\begin{align}\label{eq:middle_step_for_two_subprocess1}
  \rbra{G_{1,:-k_1} \otimes I_{d_2}}\mbox{vec}\rbra{D_{1,-k_1:}^T} &=  \rbra{G_{1,:k_1} \otimes I_{d_2}}\mbox{vec}\rbra{\bSigma_{12,0}^T}  \\
  \rbra{(G_2L_2)_{:-k_1} \otimes I_{d_1}}\mbox{vec}\rbra{(L_2D_2)_{-k_1:}^T} &=  \rbra{(G_2L_2)_{:k_1} \otimes I_{d_1}}\mbox{vec}\rbra{\bSigma_{12,0}}
\end{align}
since
\begin{equation}\label{eq:middle_step_for_two_subprocess2}
  \mbox{vec}\rbra{\bSigma_{12,i}^T} = K_{d_1d_2}\mbox{vec}\rbra{\bSigma_{12,i}}
\end{equation}
for $-k\leq i \leq k$ where the above defines for commutation matrix
$K_{d_1d_2}$ as the permutation matrix to convert from vectorization by
rows versus vectorization by columns.
It can be checked that
\begin{equation}\label{eq:middle_step_for_two_subprocess3}
  \mbox{vec}\rbra{D_{1,-k_1:}^T} = \rbra{I_{2k} \otimes K_{d_1d_2}}\mbox{vec}\rbra{(L_2D_2)_{-k_1:}^T}.
\end{equation}
Plugging \cref{eq:middle_step_for_two_subprocess3} and \cref{eq:middle_step_for_two_subprocess2} into \cref{eq:middle_step_for_two_subprocess1} leads to
\begin{multline}\label{eq:linear_eq_for_11}
  \begin{pmatrix}
    \rbra{G_{1,:-k_1} \otimes I_{d_2}}\rbra{I_{2k} \otimes K_{d_1d_2}} \\
    (G_2L_2)_{:-k_1} \otimes I_{d_1}
  \end{pmatrix}
  \mbox{vec}\sbra{\rbra{ \bSigma_{12, -k}, \dots, \bSigma_{12,-1}, \bSigma_{ij,1}, \dots, \bSigma_{ij,k}}} \\=
  \begin{pmatrix}
    \rbra{G_{1,:k_1} \otimes I_{d_2}}K_{d_1d_2} \\
    G_{2,:k_1} \otimes I_{d_1}
  \end{pmatrix}
  \mbox{vec}\rbra{\bSigma_{12,0}}.
\end{multline}
Therefore, given $G_1, G_2$, $D_1$, $D_2$ and with the fact that the first coefficient matrix of $\mbox{vec}\rbra{D_{2,-k_1:}^T}$ is non-singular in nearly all situations, $\bSigma_{12, -k}, \dots, \bSigma_{12,-1}, \bSigma_{12,1}, \dots, \bSigma_{12,k}$ can be uniquely solved when $\bSigma_{12,0}$ is fixed.

\medskip

\noindent
\textbf{Case 2}.
Submatrices of $H$ are defined similar to submatrices of $G$.
The linear systems can be written as:
\begin{multline}\label{eq:linear_eq_for_22}
  \begin{pmatrix}
    \rbra{H_{1,:-k_1} \otimes I_{d_2}}\rbra{I_{2k} \otimes K_{d_1d_2}} \\
    (H_2L_2)_{:-k_1} \otimes I_{d_1}
  \end{pmatrix}
  \mbox{vec}\sbra{\rbra{ \bSigma_{12, -k}, \dots, \bSigma_{12,-1}, \bSigma_{12,1}, \dots, \bSigma_{12,k}}} \\=
  \begin{pmatrix}
    \rbra{H_{1,:k_1} \otimes I_{d_2}}K_{d_1d_2} \\
   H_{2,:k_1} \otimes I_{d_1}
  \end{pmatrix}
  \mbox{vec}\rbra{\bSigma_{12,0}},
\end{multline}
and is similar to \cref{eq:linear_eq_for_11}.

\subsection{Multiple sub-processes}\label{appd:multi_sub_processes}

The $\mbox{vec}\rbra{\bSigma_{12,0}}$ form of the linear systems
are given in  this subection for the case of 2 or more sub-processes. Similar to \cref{eq:linear_eq_for_11} and \cref{eq:linear_eq_for_22}, given $\bSigma_{ii,0},\bSigma_{jj,0},\dots,\bSigma_{ii,k}, \bSigma_{jj,k}$, the equations of four different cases of $(c_i,c_j)$ can be concluded as following:
\begin{enumerate}
\item $(c_i,c_j) = (1,1)$: set $\bSigma_{ij,0}$ as fixed parameter and $\bSigma_{ij, -k}, \dots, \bSigma_{ij,-1}, \bSigma_{ij,1}, \dots, \bSigma_{ij,k}$ can be obtained by solving the linear equation
    \begin{multline}\label{eq:linear_eq_for_SiSj_11}
        \mbox{vec}\sbra{\rbra{ \bSigma_{ij, -k}, \dots, \bSigma_{ij,-1}, \bSigma_{ij,1}, \dots, \bSigma_{ij,k} }} = \\
            \begin{pmatrix}
                \rbra{G_{i,:-k_1} \otimes I_{d_j}}\rbra{I_{2k} \otimes K_{d_id_j}} \\
                \rbra{G_jL_j}_{:-k_1} \otimes I_{d_i}
            \end{pmatrix}^{-1}
            \begin{pmatrix}
                \rbra{G_{i,:k_1} \otimes I_{d_j}}K_{d_id_j} \\
                G_{j,:k_1} \otimes I_{d_j}
            \end{pmatrix}
        \mbox{vec}\rbra{\bSigma_{ij,0}};
    \end{multline}
\item $(c_i,c_j) = (1,2)$: set $\bSigma_{ij,-k}$ as fixed parameter and $\bSigma_{ij, -k+1} = \cdots = \bSigma_{ij,k} = \bm{0}$;
\item $(c_i,c_j) = (2,1)$: set $\bSigma_{ij,k}$ as fixed parameter and $\bSigma_{ij, -k} = \cdots = \bSigma_{ij,k-1} = \bm{0}$.
\item $(c_i,c_j) = (2,2)$: set $\bSigma_{ij,0}$ as fixed parameter and $\bSigma_{ij, -k}, \dots, \bSigma_{ij,-1}, \bSigma_{ij,1}, \dots, \bSigma_{ij,k}$ can be obtained by solving the linear equation
    \begin{multline}\label{eq:linear_eq_for_SiSj_22}
        \mbox{vec}\sbra{\rbra{ \bSigma_{ij, -k}, \dots, \bSigma_{ij,-1}, \bSigma_{ij,1}, \dots, \bSigma_{ij,k} }} = \\
            \begin{pmatrix}
                \rbra{H_{i,:-k_1} \otimes I_{d_j}}\rbra{I_{2k} \otimes K_{d_id_j}} \\
                \rbra{H_jL_j}_{:-k_1} \otimes I_{d_i}
            \end{pmatrix}^{-1}
            \begin{pmatrix}
                \rbra{H_{i,:k_1} \otimes I_{d_j}}K_{d_id_j} \\
                H_{j,:k_1} \otimes I_{d_i}
            \end{pmatrix}
        \mbox{vec}\rbra{\bSigma_{ij,0}}.
    \end{multline}
\end{enumerate}

\end{document}